\documentclass[acmsmall,screen]{acmart}

\setcopyright{cc}
\setcctype{by}
\acmDOI{10.1145/3776673}
\acmYear{2026}
\acmJournal{PACMPL}
\acmVolume{10}
\acmNumber{POPL}
\acmArticle{31}
\acmMonth{1}
\received{2025-07-10}
\received[accepted]{2025-11-06}

\usepackage{main}
\usepackage{widows-and-orphans}

\begin{document}

\title{Abstraction Functions as Types}
\subtitle{Modular Verification of Cost and Behavior in Dependent Type Theory}

\begin{abstract}
  Software development depends on the use of libraries whose public specifications inform client code and impose obligations on private implementations; it follows that verification at scale must also be modular, preserving such abstraction.
  Hoare's influential methodology uses \emph{abstraction functions} to demonstrate the coherence between such concrete implementations and their abstract specifications.
  However, the Hoare methodology relies on a conventional separation between implementation and specification, providing no linguistic support for ensuring that this convention is obeyed.

  This paper proposes a synthetic account of Hoare's methodology within univalent dependent type theory by encoding the data of abstraction functions within types themselves.
  This is achieved via a \emph{phase distinction}, which gives rise to a \emph{gluing} construction that renders an abstraction function as a type and a pair of modalities that \emph{fracture} a type into its concrete and abstract parts.
  A \emph{noninterference} theorem governing the phase distinction characterizes the modularity guarantees provided by the theory.

  This approach scales to verification of cost, allowing the analysis of client cost relative to a cost-aware specification.
  A monadic \emph{sealing} effect facilitates modularity of cost, permitting an implementation to be upper-bounded by its specification in cases where private details influence observable cost.
  The resulting theory supports modular development of programs and proofs in a manner that hides private details of no concern to clients while permitting precise specifications of both the cost and behavior of programs.
\end{abstract}

\author{Harrison Grodin}
\orcid{0000-0002-0947-3520}
\email{hgrodin@cs.cmu.edu}

\author{Runming Li}
\orcid{0000-0001-7600-9069}
\email{runmingl@cs.cmu.edu}

\author{Robert Harper}
\orcid{0000-0002-9400-2941}
\email{rwh@cs.cmu.edu}

\affiliation{
  \institution{Carnegie Mellon University}
  \department{Computer Science Department}
  \city{Pittsburgh}
  \state{PA}
  \country{USA}
}



\maketitle

\section{Introduction}\label{sec:introduction}

The single most effective tool in software development is the composition of programs from reusable, replaceable parts.
Developers employ abstract types to describe the assumptions on which clients may rely when using a component, and, correspondingly, the obligations that implementers must provide.
The essence of modularity is that such clients only receive a limited understanding of the real code, thus allowing implementations validating the provided mental model to be freely swapped out for one another without affecting the behavior of the surrounding programs.

The verification of abstract types has classically been achieved through the use of \emph{abstraction functions}, functions that translate a concrete data structure state to a corresponding mathematical model~\citep{hoare>1972}.
Every concrete type is equipped with an abstract model and an abstraction function, and all operations must preserve this abstraction function.
For example, an interface describing first-in-first-out queues might include a type component and the following operations:
\[ \sum_{X : \tpv} (\Label{empty} : X) \times (\Label{enqueue} : E \to X \to X) \times (\Label{dequeue} : X \to E \times X) \]
Mathematically, this interface may be inhabited by the abstract specification
\[ ({X}_{\abs}, \Impl{empty}_{\abs}, \Impl{enqueue}_{\abs}, \Impl{dequeue}_{\abs}) \isdef (\listty{E}, \nilex, \lambda~e~l.\ \catlist{l}{\singex{e}}, \Impl{uncons}) \]
in terms of lists.
However, queues may be implemented otherwise, improving on the efficiency of this na\"ive specification.
Verification that such an implementation $({X}_{\top}, \Impl{empty}_{\top}, \Impl{enqueue}_{\top}, \Impl{dequeue}_{\top})$ is correct relative to the specification requires a function from the implementation type $X_{\top}$ to the specification type ${X}_{\abs}$ that coheres with each implemented operation $\TOP{f}$:
\begin{myquote}{hoare>1972}
  The first requirement for the proof\dots is to define the relationship between the abstract space [${X}_{\abs}$] in which the abstract program [${f}_{\abs}$] is written, and the space [${X}_{\top}$] of the concrete representation.
  This can be accomplished by giving a function [$\alpha$] which maps the concrete variables into the abstract object which they represent\dots
\end{myquote}
First, the type component induces the obligation of an abstraction function $\alpha : X_{\top} \to X_{\abs}$, converting from the implementation type $X_{\top}$ to the specification type $X_{\abs}$ (here, $\listty{E}$).
Then, for each operation, the implementation must match the list specification relative to the function $\alpha$.
{\small
\[\begin{tikzcd}
	1 & {X_{\top}} && {X_{\top}} && {X_{\top}} && {X_{\top}} && {E \times X_{\top}} \\
	1 & {X_{\abs}} && {X_{\abs}} && {X_{\abs}} && {X_{\abs}} && {E \times X_{\abs}}
	\arrow["{\Impl{empty}_{\top}}", from=1-1, to=1-2]
	\arrow[r,-,double equal sign distance,double, from=1-1, to=2-1]
	\arrow["\alpha", from=1-2, to=2-2]
	\arrow["{\Impl{enqueue}_{\top}~e}", from=1-4, to=1-6]
	\arrow["{\alpha}"', from=1-4, to=2-4]
	\arrow["\alpha", from=1-6, to=2-6]
	\arrow["{\Impl{dequeue}_{\top}}", from=1-8, to=1-10]
	\arrow["\alpha"', from=1-8, to=2-8]
	\arrow["E \times \alpha", from=1-10, to=2-10]
	\arrow["{\Impl{empty}_{\abs}}"', from=2-1, to=2-2]
	\arrow["{\Impl{enqueue}_{\abs}~e}"', from=2-4, to=2-6]
	\arrow["{\Impl{dequeue}_{\abs}}"', from=2-8, to=2-10]
\end{tikzcd}\]}%
Such coherence conditions compose horizontally, guaranteeing that client code using the implementation type $X_{\top}$ through the interface only will preserve the required abstraction functions.

Analytically, much work has been done on verifying data structures using abstraction functions.
However, such formalizations provide no intrinsic guarantee of modularity: it is possible to write a function that uses a private property of the queue representation type $X_{\top}$ not revealed by the list representation type $X_{\abs}$.
The onus is on the client to check that its own correctness relies only on publicly-available data by proving the appropriate coherence squares; otherwise, modularity would be violated, as a replacement implementation could not be substituted while guaranteeing that client behavior is unaffected.
To avoid such issues, we equip every concrete type with an abstract model and an abstraction function and every concrete implementation program with a coherent abstract specification program, achieved using the \emph{synthetic} approach of a phase distinction.

\subsection{Synthetic Phase Distinctions}

A synthetic phase distinction ensures that every type $X$ consists of a concrete type $\ALGO{X}$, an abstract type $\BEH{X}$, and an abstraction function $\chi$.
Correspondingly, every function $f : X \to Y$ will contain a concrete implementation $\ALGO{f} : \ALGO{X} \to \ALGO{Y}$, an abstract specification $\BEH{f} : \BEH{X} \to \BEH{Y}$, and a proof of their coherence relative to the bounding abstraction functions.
Phase distinctions first arose in module calculi~\citep{harper-mitchell-moggi>1989}, where a phase isolates compile-time data from run-time data.
By analogy, modularity in dependent type theory requires abstract data, for mathematical specification and reasoning, to be isolated from concrete data, for efficient implementation of programs.

Synthetically, a \emph{phase} is a proposition that, when assumed to be true, isolates a desired aspect of all relevant types and programs~\citep{sterling-harper>2021}.
For the isolation of abstract data from concrete implementation details, we introduce the \emph{abstract phase}, $\abs$, under which private data is hidden.
Every phase induces two idempotent monadic modalities for isolating the in-phase and out-of-phase aspects of a type.
In particular, the \emph{abstract modality} $\Op_{\abs}$ isolates the mathematical model contained within a type, and the \emph{concrete modality} $\Cl_{\abs}$ isolates the implementation type contained within a type.
Moreover, in univalent type theory~\citep{univalentfoundations>2013}, a \emph{gluing} construction builds a type from a function from a concrete type to an abstract type~\citep[\S 3.4]{rijke-shulman-spitters>2020} that we will treat as an abstraction function.

\subsection{Modular Verification of Cost}

The proposed synthetic viewpoint on abstraction supports modular specification not only of behavior, but also of cost.
Using Calf~\citep{niu-sterling-grodin-harper>2022}, a dependent type theory that reifies cost as an effect, programs are instrumented with cost annotations that reflect an abstract notion of resource usage in a program.
This allows the abstract aspect of a program to propagate cost information to client code, facilitating the modular verification of complexity as has been a longstanding goal:
\begin{myquote}{stepanov>1995}
  It was commonly assumed that the complexity of an operation is part of implementation and that abstraction ignores complexity\dots
  Complexity\dots
  is a part of the unwritten contract between the module and its user\dots
  You cannot have interchangeable modules unless these modules share similar complexity behavior\dots
  Complexity assertions have to be part of the interface.
\end{myquote}
The inclusion of cost guarantees in interfaces is essential in the modular design of algorithms and data structures, allowing downstream clients to use a simplified model of the cost characteristics of libraries when verifying downstream code~\citep{acar-blelloch>2022}.

In supporting cost in interfaces, the ability to make statements about program behavior only is lost.
To counteract this effect, the behavioral phase $\beh$ of Calf is incorporated, suppressing cost annotations to isolate behavior.
In tandem the abstract and behavioral phases smoothly transition from implementation, to cost specification, to behavior specification.

Finally, cost specifications should be upper bounds, rather than exact requirements, on implementations.
Not only do looser specifications allow for greater flexibility across implementations, but often, the true cost of a program depends on concrete details hidden in the abstract phase, mandating that merely an upper bound be satisfied to preserve abstraction.
Building on the Decalf extension for cost reasoning with inequality~\citep{grodin-niu-sterling-harper>2024}, a sealing-style effect is developed, allowing a concrete algorithm to be hidden under an abstract upper bound it purports to meet.

\subsection{Contributions}\label{sec:contributions}

This work applies the modal framework of \citet{rijke-shulman-spitters>2020} to the problem of modular programming and verification in \cref{sec:abstraction}, using the proposed abstract phase to build into each type both a concrete and an abstract part, connected by an abstraction function.
The noninterference theorem associated with this phase mediates the boundary between components in \cref{sec:interface}, allowing concrete library implementations to be chosen arbitrarily without impacting downstream code or verifications.
In \cref{sec:cost} the cost effect and behavioral phase of Calf are integrated to refine the type theory further, providing both cost-aware and cost-free interfaces; finally, combining the abstract phase with the inequational reasoning of Decalf, a monadic sealing effect is introduced to facilitate modular cost verification using upper bounds.

\section{Abstraction Functions, Synthetically}\label{sec:abstraction}

To fracture dependent type theory into concrete and abstract aspects, only a single proposition $\abs$, the \emph{abstract phase}, must be added as a primitive.
Because this phase is at most inhabited, the variable name $\b : \abs$ is used by convention to assume its truth.

\subsection{The Abstract Phase}

We first recall fundamental definitions and results from \citet{rijke-shulman-spitters>2020}.
Associated with the phase $\abs$ are a pair of idempotent monadic modalities, the abstract modality $\Op_{\abs}$ and the concrete modality $\Cl_{\abs}$, that will isolate the abstract and concrete aspects of a given type, respectively.%
\footnote{\citet{rijke-shulman-spitters>2020} refer to $\Op_{\abs}$ and $\Cl_{\abs}$ as the open and closed modalities determined by $\abs$, respectively.}

\begin{definition}\label{def:open-modality}
  The \emph{abstract modality} $\Op_{\abs} X \isdef \abs \to X$ is the reader monad for the type $\abs$.
  A type $X$ is \emph{abstract} when the monad unit $\etaOp_X : X \to \Op_{\abs} X$ is an equivalence%
  \footnote{To show that a function $f : X \to Y$ is an \emph{equivalence}, it suffices to show that it is an isomorphism by giving an inverse function $f^{-1} : Y \to X$ where both $f^{-1} \circ f = \text{id}_{X}$ and $f \circ f^{-1} = \text{id}_{Y}$~\citep{univalentfoundations>2013}.};
  an inhabitant of an abstract type is called a \emph{specification}.
  For convenience, let $\eqOp{x}{x'}$ be a shorthand for $\Op_{\abs}(x = x')$.
\end{definition}

\begin{definition}\label{def:closed-modality}
  The \emph{concrete modality} $\Cl_{\abs} X$ marks a type as irrelevant in the abstract phase: crucially, given $\b : \abs$, we have $\Cl_{\abs} X \iso \unit$.
  It is defined as the following quotient-inductive type (a pushout), including injections $\etaCl_X$ for $X$ and $\starCl$ for $\abs$ that are identified in the phase:
  \begin{center}
    \begin{minipage}{0.5\linewidth}
      \begin{center}
      \iblock{
        \mhang{\kw{data}~\Cl_{\abs}~(X : \tpv) : \tpv~\kw{where}}{
          \mrow{\Label{\etaCl_X} : X \to \Cl_{\abs} X}
          \mrow{\Label{\starCl} : \abs \to \Cl_{\abs} X}
          \mrow{\Label{\_} : (x : X)~(\b : \abs) \to \Label{\etaCl_X} x = \Label{\starCl~\b}}
        }
      }
      \end{center}
    \end{minipage}%
    \begin{minipage}{0.5\linewidth}
      \[\begin{tikzcd} {\abs \times X} & \abs \\
        X & {\Cl_{\abs} X}
        \arrow["{\Label{proj}_2}"', from=1-1, to=2-1]
        \arrow["{\Label{proj}_1}", from=1-1, to=1-2]
        \arrow["{\starCl}", from=1-2, to=2-2]
        \arrow["{\etaCl_X}"', from=2-1, to=2-2]
        \arrow["\lrcorner"{anchor=center, pos=0.125, rotate=180}, draw=none, from=2-2, to=1-1]
      \end{tikzcd}\]
    \end{minipage}%
  \end{center}
  The quotient must be respected when pattern matching on this modality: when casing on concrete data of type $\Cl_{\abs} X$, both the $\etaCl_X$ and $\starCl$ cases must agree (in the abstract phase, because $\starCl$ may only be constructed assuming $\abs$ is inhabited).
  A type $X$ is \emph{concrete} when the monad unit $\etaCl_X : X \to \Cl_{\abs} X$ is an equivalence; an inhabitant of a concrete type is an \emph{implementation}.
\end{definition}

The concrete modality indicates that data is present for implementation purposes only, to be ignored abstractly.
Concrete types can be characterized in terms of the abstract modality:

\begin{lemma}\label{lem:closed-modal}
  A type $X$ is concrete exactly when $\Op_{\abs} X$ is contractible (equivalent to $1$).
\end{lemma}

\subsubsection{Semantics}\label{sec:semantics}

There are many possible semantics for such a phased language, differing centrally in their interpretation of the phase proposition, $\abs$.

\begin{model}\label{mod:true}
  The phase $\abs$ may be interpreted as true, causing $\Op_{\abs} X = X$ and $\Cl_{\abs} X = 1$.
\end{model}

\begin{model}\label{mod:false}
  The phase $\abs$ may be interpreted as false, causing $\Op_{\abs} X = 1$ and $\Cl_{\abs} X = X$.
\end{model}

The in-phase semantics extracts the abstract mathematical specification, and the out-of-phase semantics extracts the concrete underlying code; both ignore abstraction functions.
In a third semantics the proposition $\abs$ is an intermediate ``included middle'' proposition (in a constructive setting) that is neither true nor false; this model interprets every type as its abstraction function.

\begin{model}\label{mod:psh}
  Using Kripke semantics with worlds as phases ordered by entailment,\footnote{That is, presheaves on the poset $\vmathbb{2} \isdef \{ \abs \vdash \top \}$ of propositions~\cite[\S 5]{niu-sterling-grodin-harper>2022}.}
  each type $X$ is interpreted as
  an out-of-phase component $\interp{X}_\top$, an in-phase component $\interp{X}_{\abs}$, and an abstraction function $\interp{X}_\top \to \interp{X}_{\abs}$.
  For example:
  \begin{align*}
    \interp{X} &\isdef \Presheaf{\interp{X}_\top}{\interp{X}_{\abs}} &
    \interp{\nat} &\isdef \Presheaf[\text{id}]{\nat}{\nat} &
    \interp{\abs} &\isdef \Presheaf{0}{1} = \Yo(\abs)
  \end{align*}
  Standard base types, such as $\nat$, are interpreted with identical concrete and abstract components and the identity abstraction function (that is, as constant presheaves).
  The abstract phase is defined to have an empty concrete component and a trivial abstract component (that is, as the Yoneda embedding of $\abs : \vmathbb{2}$), combining \cref{mod:true,mod:false}.
  The modalities are interpreted as follows:
  \begin{align*}
    \interp{\Op_{\abs} X} &= \interp{\abs \to X} = \Presheaf[\text{id}]{\interp{X}_{\abs}}{\interp{X}_{\abs}} &
    \interp{\Cl_{\abs} X} &= \Presheaf{\interp{X}_\top}{1}
  \end{align*}
  The abstract modality treats abstract data as \emph{both} specification and implementation, using the identity function as an abstraction function.
  The concrete modality erases abstract data, replacing the abstract component with the singleton type $1$ and inducing the unique abstraction function into the singleton.
  Notice that a type $X$ is abstract when it is interpreted with $\interp{X}_\top = \interp{X}_{\abs}$ and the identity abstraction function, and it is concrete when $\interp{X}_{\abs} = 1$.
\end{model}

\subsubsection{Fracture and Gluing}

An abstract and a concrete type may be combined to form a single type: given a concrete type $\ALGO{X}$ and an abstract type $\BEH{X}$, we wish to construct a type $X$ with $\Cl_{\abs} X = \ALGO{X}$ and $\Op_{\abs} X = \BEH{X}$.
The product $\ALGO{X} \times \BEH{X}$ already satisfies the second criterion, since abstractly $\ALGO{X} = 1$, but it does not satisfy the first: abstract types themselves have nontrivial concrete aspects, so $\Cl_{\abs} \BEH{X}$ does not trivialize.
This may be resolved by restricting the product $\ALGO{X} \times \BEH{X}$ such that the components concretely cohere according to a phase-sensitive analogue of an abstraction function.

\begin{definition}
  A \emph{phased abstraction function} from $\ALGO{X}$ to $\BEH{X}$ is a function $\chi : \ALGO{X} \to \Cl_{\abs} \BEH{X}$.
\end{definition}

The modality on the output lifts $\BEH{X}$ to the level of concrete code, so a phased abstraction function is a function between concrete types, trivialized in the abstract phase.%
\footnote{Without the concrete modality on the output, a function $\ALGO{X} \to \BEH{X}$ would necessarily be constant, per \cref{thm:noninterference}.}
Now, a concrete type $\ALGO{X}$, an abstract type $\BEH{X}$, and a phased abstraction function $\chi$ may be \emph{glued} together:%
\footnote{This presents the universe $\tpv$ as a synthetic Artin gluing~\citep[\S 3.4]{rijke-shulman-spitters>2020}.}
\begin{center}
  \begin{minipage}{0.6\textwidth}
    \begin{align*}
      & \mathsf{Glue} : \left( \sum_{\ALGO{X} : \ALGO{\tpv}} \sum_{\BEH{X} : \BEH{\tpv}} (\ALGO{X} \to \Cl_{\abs} \BEH{X}) \right) \to \tpv \\
      & \Glue{\BEH{X}}{\ALGO{X}}{\chi} \isdef \sum_{(\ALGO{x}, \BEH{x}) : \ALGO{X} \times \BEH{X}} \chi \ALGO{x} =_{\Cl_{\abs} \BEH{X}} \etaCl_{\BEH{X}} \BEH{x}
    \end{align*}
  \end{minipage}%
  \begin{minipage}{0.4\textwidth}
    \[\begin{tikzcd} \Glue{\BEH{X}}{\ALGO{X}}{\chi} & \ALGO{X} \\
      \BEH{X} & {\Cl_{\abs} \BEH{X}}
      \arrow["\Label{proj}_2"', from=1-1, to=2-1]
      \arrow["\Label{proj}_1", from=1-1, to=1-2]
      \arrow["{\chi}", from=1-2, to=2-2]
      \arrow["{\etaCl_{\BEH{X}}}"', from=2-1, to=2-2]
      \arrow["\lrcorner"{anchor=center, pos=0.125}, draw=none, from=1-1, to=2-2]
    \end{tikzcd}\]
  \end{minipage}%
\end{center}
Here, $\ALGO{\tpv}$ and $\BEH{\tpv}$ are the sub-universes of concrete and abstract types, respectively.
Gluing $\ALGO{X}$ to $\BEH{X}$ along $\chi$ restricts the product $\ALGO{X} \times \BEH{X}$ to only include pairs $(\ALGO{x}, \BEH{x})$ that concretely cohere up to the phased abstraction function $\chi$, with $\chi \ALGO{x} = \etaCl \BEH{x}$.
In fact, every type $X$ contains a phased abstraction function $\Cl_{\abs} \etaOp_X : \Cl_{\abs} X \to \Cl_{\abs} \Op_{\abs} X$ from $\Cl_{\abs} X$ to $\Op_{\abs} X$, which can be extracted via a map in the reverse direction called \emph{fracturing}:
\begin{align*}
  &\mathsf{Frac} : \tpv \to \left( \sum_{\ALGO{X} : \ALGO{\tpv}} \sum_{\BEH{X} : \BEH{\tpv}} (\ALGO{X} \to \Cl_{\abs} \BEH{X}) \right) \\
  &\Frac{X} \isdef (\Cl_{\abs} X, \Op_{\abs} X, \Cl_{\abs} X \xrightarrow{\Cl_{\abs} \etaOp_X} \Cl_{\abs} \Op_{\abs} X)
\end{align*}
This fracturing construction is inverse to gluing.

\begin{samepage}
\begin{theorem}[Fracture and Gluing]\label{thm:fracture}
  $\mathsf{Frac}$ and $\mathsf{Glue}$ form an equivalence~\citep[\S 3.4]{rijke-shulman-spitters>2020}.
\end{theorem}
The round-trip condition $\mathsf{Glue}(\Frac{X}) = X$ on $\tpv$ says that every type $X : \tpv$ can be rendered as a gluing of its fractured components $\Cl_{\abs} X$ and $\Op_{\abs} X$:
\begin{center}
  \begin{minipage}{0.6\linewidth}
    \[
      X = \left( \sum_{(\ALGO{x}, \BEH{x}) : \Cl_{\abs} X \times \Op_{\abs} X} (\Cl_{\abs} \etaOp_X) \ALGO{x} =_{\Cl_{\abs} \Op_{\abs} X} \etaCl_{\Op_{\abs} X} \BEH{x} \right)
    \]
  \end{minipage}%
  \begin{minipage}{0.4\linewidth}
    \[\begin{tikzcd} X & \Cl_{\abs} X \\
      \Op_{\abs} X & {\Cl_{\abs} \Op_{\abs} X}
      \arrow["{\etaOp_X}"', from=1-1, to=2-1]
      \arrow["{\etaCl_X}", from=1-1, to=1-2]
      \arrow["{\Cl_{\abs} \etaOp_X}", from=1-2, to=2-2]
      \arrow["{\etaCl_{\Op_{\abs} X}}"', from=2-1, to=2-2]
      \arrow["\lrcorner"{anchor=center, pos=0.125}, draw=none, from=1-1, to=2-2]
    \end{tikzcd}\]
  \end{minipage}%
\end{center}
The other round-trip condition ensures $\Cl_{\abs} \Glue{\BEH{X}}{\ALGO{X}}{\chi} = \ALGO{X}$ and $\Op_{\abs} \Glue{\BEH{X}}{\ALGO{X}}{\chi} = \BEH{X}$.
\end{samepage}

\subsection{Abstraction Functions Explicitly, via Gluing}\label{sec:abstraction:gluing}

Gluing may be used to build a type out of any abstraction function, including analytic verifications using abstraction functions into the synthetic setting.
We now consider two case studies of concrete types implementing the list abstract type, $\BEH{X} \isdef \Op_{\abs} (\listty{E})$.

\subsubsection{Batched Queue}\label{sec:abstraction:gluing:queue}

To implement queues within a functional setting with efficient amortized cost, a pair of lists may be used: incoming data is enqueued to the ``inbox'' list, and outgoing data is usually dequeued from the ``outbox'' list, unless it is empty, in which case the ``inbox'' data is moved to the ``outbox'' as a batch~\citep{hood-melville>1981,burton>1982,gries>1989,okasaki>1999}.
Let $E : \tpv$ be the type of elements to store, equipped with a default element to dequeue from an empty queue for simplicity.
To implement a queue in this way we may use gluing, selecting the concrete representation type $\ALGO{X}$ as pairs of lists and a phased abstraction function $\chi : \ALGO{X} \to \Cl_{\abs} \BEH{X}$:
\begin{align*}
  \ALGO{X} &\isdef \Cl_{\abs} (\listty{E} \times \listty{E}) &
  \BEH{X} &\isdef \Op_{\abs} (\listty{E}) &
  \chi &\isdef \Cl_{\abs} (\etaOp \circ \Impl{revAppend})
\end{align*}
where $\Impl{revAppend}~(l_1, l_2) \isdef \catlist{l_2}{\rev{l_1}}$. Encoded within the type $X \isdef \Glue{\BEH{X}}{\ALGO{X}}{\chi}$ itself is the idea that every pair of lists can be transformed to its abstract list representation by appending the reversed inbox list to the outbox. The goal is to give an implementation of the following queue signature whose representation type $X$ is $\Glue{\BEH{X}}{\ALGO{X}}{\chi}$:
\[ \Name{PreQueue} \isdef \sum_{X : \tpv} (\Label{empty} : X) \times (\Label{enqueue} : E \to X \to X) \times (\Label{dequeue} : X \to E \times X). \]

Writing code involving the type $X$ involves writing an abstract specification-level program on $\BEH{X}$, a concrete implementation on $\ALGO{X}$, and a proof that they cohere up to $\chi$.
For example, let
\begin{align*}
  \TOP{\Impl{empty}} &: \listty{E} \times \listty{E} &
  \ABS{\Impl{empty}} &: \listty{E} \\
  \TOP{\Impl{empty}} &\isdef (\nilex, \nilex) &
  \ABS{\Impl{empty}} &\isdef \nilex
\end{align*}
define the standard concrete and abstract representations of the empty queue.
Classically, we would prove that $\TOP{\Impl{empty}}$ accurately implements $\ABS{\Impl{empty}}$ by proving that the following square commutes:
\begin{center}
  \begin{minipage}{0.5\linewidth}
    \begin{align*}
      \Impl{revAppend}~\TOP{\Impl{empty}} = \ABS{\Impl{empty}}
    \end{align*}
  \end{minipage}%
  \begin{minipage}{0.5\linewidth}
\[\begin{tikzcd}
	1 & {\listty{E} \times \listty{E}} \\
	1 & {\listty{E}}
	\arrow["{\TOP{\Impl{empty}}}", from=1-1, to=1-2]
	\arrow[r,-,double equal sign distance,double, from=1-1, to=2-1]
	\arrow["\Impl{revAppend}", from=1-2, to=2-2]
	\arrow["{\ABS{\Impl{empty}}}"', from=2-1, to=2-2]
\end{tikzcd}\]
  \end{minipage}%
\end{center}
To define the empty queue of glue type $X$ using the language of phases, we include all of this data:
\begin{align*}
  &\Impl{empty} : X \\
  &\Impl{empty} \isdef (\etaCl\TOP{\Impl{empty}}, \etaOp \ABS{\Impl{empty}}, \eqrefl_{\Cl_{\abs} \BEH{X}})
\end{align*}
The abstract specification is wrapped with $\etaOp$, the concrete implementation is wrapped with $\etaCl$, and the proof $\eqrefl_{\Cl_{\abs} \BEH{X}}$ is the fact that these data cohere up to $\chi$, a phased analogue of the commutative square shown above.
(For readability, we will henceforth omit proofs of coherence by $\chi$ in writing.)

Implementation of the enqueue operation is similar, but because it takes an element of $X$ as input, we use the functorial action of the modalities. %
Given concrete and abstract programs,
\begin{align*}
  &\TOP{\Impl{enqueue}} : E \to (\listty{E} \times \listty{E}) \to (\listty{E} \times \listty{E}) &
  &\ABS{\Impl{enqueue}} : E \to \listty{E} \to \listty{E} \\
  &\TOP{\Impl{enqueue}}~e~(l_1, l_2) \isdef (\consex{e}{l_1}, l_2) &
  &\ABS{\Impl{enqueue}}~e~l \isdef \catlist{l}{\singex{e}},
\end{align*}
and a proof $\Impl{enqueue-correct}$ that they preserve $\Impl{revAppend}$,
\begin{center}
  \begin{minipage}{0.55\linewidth}
    \begin{align*}
      \Impl{revAppend} \circ \TOP{\Impl{enqueue}}~e = \ABS{\Impl{enqueue}}~e \circ \Impl{revAppend}
    \end{align*}
  \end{minipage}%
  \begin{minipage}{0.45\linewidth}
    \[\begin{tikzcd}
    	{\listty{E} \times \listty{E}} & {\listty{E} \times \listty{E}} \\
    	{\listty{E}} & {\listty{E}}
    	\arrow["{\TOP{\Impl{enqueue}}~e}", from=1-1, to=1-2]
    	\arrow["{\Impl{revAppend}}"', from=1-1, to=2-1]
    	\arrow["{\Impl{revAppend}}", from=1-2, to=2-2]
    	\arrow["{\ABS{\Impl{enqueue}}~e}"', from=2-1, to=2-2]
    \end{tikzcd},\]
  \end{minipage}%
\end{center}
we can implement the enqueue operation on $X$ using the action of the modalities $\Cl_{\abs}$ and $\Op_{\abs}$,
\begin{align*}
  &\Impl{enqueue} : E \to X \to X \\
  &\Impl{enqueue}~e \isdef \Cl_{\abs} (\TOP{\Impl{enqueue}}~e) \times \Op_{\abs} (\ABS{\Impl{enqueue}}~e).
\end{align*}
The proof obligation induced by the gluing construction that the concrete and abstract parts of the enqueue operation cohere under the abstraction function $\chi$ is stated as follows:
\begin{align*}
  (\ALGO{x} : \ALGO{X})~(\BEH{x} : \BEH{X})
  \to \chi(\ALGO{x}) = \etaCl(\BEH{x})
  \to \chi(\Cl_{\abs}(\TOP{\Impl{enqueue}}~e)~\ALGO{x}) = \etaCl(\Op_{\abs}(\ABS{\Impl{enqueue}}~e)~\BEH{x})
\end{align*}
This follows from $\Impl{enqueue-correct}$, the above coherence square.
In this manner an analytic, phase-insensitive codebase such as $\TOP{\Impl{enqueue}}$, $\ABS{\Impl{enqueue}}$, and $\Impl{enqueue-correct}$ can be systematically and mechanically converted into synthetic, phase-sensitive code, such as $\Impl{enqueue}$.
Finally, we combine
\begin{align*}
  &\TOP{\Impl{dequeue}} : (\listty{E} \times \listty{E}) \to E \times (\listty{E} \times \listty{E}) &
  &\ABS{\Impl{dequeue}} : \listty{E} \to E \times \listty{E}
\end{align*}
to define $\Impl{dequeue} : X \to E \times X$.
While the second component is analogous to $\Impl{enqueue}$, the first component (of type $E$) requires additional care.
Using the action of the modalities, the function
\begin{align*}
  \Cl_{\abs}(\mathsf{proj}_1 \circ \TOP{\Impl{dequeue}}) \times \Op_{\abs}(\mathsf{proj}_1 \circ \ABS{\Impl{dequeue}}) &: X \to \Glue{\Op_{\abs}E}{\Cl_{\abs}E}{\Cl_{\abs}\etaOp_E}
\end{align*}
does not map into $E$, but rather its fractured-and-glued form.
To remedy this, this function is post-composed with the equivalence $\mathsf{Glue}(\Frac{E}) \to E$ induced by \cref{thm:fracture}.

\begin{remark}[Semantics]
  The various semantics isolate different aspects of the type $X$.
  \begin{enumerate}
    \item
      In \cref{mod:true} the representation type $X$ is simply $\listty{E}$, retaining the abstract requirement for the implementation type although obliterating the pair-of-lists split.
    \item
      In \cref{mod:false} the representation type $X$ is $\listty{E} \times \listty{E}$, recovering the usual functional implementation of batched queues, because the single-list representation is now hidden under an impossible assumption.
    \item
      In \cref{mod:psh} the type $X$ maintains both the ``concrete'' batched and ``abstract'' list representations, as well as the abstraction function $\Impl{revAppend}$, all within a single glue type:
      \begin{align*}
        \interp{X} \isdef \interp{\Glue{\Op_{\abs}(\listty{E})}{\Cl_{\abs}(\listty{E} \times \listty{E})}{\chi}}
        = \Presheaf[\Impl{revAppend}]{\listty{E} \times \listty{E}}{\listty{E}}
      \end{align*}
      Functions will be interpreted as squares, witnessing the coherence between the concrete implementation and the abstract specification.\qedhere
  \end{enumerate}
\end{remark}

\subsubsection{Red-Black Tree}

In addition to queues the list abstract data type can represent finite, ordered \emph{sequences} of data~\cite{blelloch>1992}.
For efficiency of access to arbitrary elements, this abstract type can be implemented using a concrete type of tree equipped with some additional data to maintain approximate balance, such as a red-black tree \cite{guibas-sedgewick>1978,okasaki>1999}.%
\footnote{Commonly, red-black trees store data at the internal $\Label{red}$ and $\Label{black}$ nodes. However, we choose to store data at leaves, structuring the tree like the free monoid on $E$ because we abstractly view red-black trees as sequences.}
Following \citet{weirich>2014}, we may define an indexed inductive type to enforce the red-black invariants, with indices for tree color and black-height, as shown in \cref{fig:rbt}.
\begin{figure}
  \iblock{
    \mhang{\kw{data}~\Type{IRBTree}~ (E : \tpv) : (c : \colorty) \to (n : \nat) \to \tpv~\kw{where}}{
      \mrow{\Label{leaf} : E \to \irbtty{\black}{\zero}{E}}
      \mrow{\Label{empty} : \irbtty{\black}{\zero}{E}}
      \mrow{\Label{red} : \irbtty{\black}{n}{E} \to \irbtty{\black}{n}{E} \to \irbtty{\red}{n}{E}}
      \mrow{\Label{black} : \irbtty{c_1}{n}{E} \to \irbtty{c_2}{n}{E} \to \irbtty{\black}{(\suc{n})}{E}}
    }
    \mrow{}
    \mrow{\rbtty{E} \isdef \sum_{\{c : \colorty\}} \sum_{\{n : \nat\}} \irbtty{c}{n}{E}}
  }
  \caption{Type representing invariant-preserving red-black trees.}\label{fig:rbt}
  \Description{Purely textual code figure.}
\end{figure}

\begin{remark}[Representation Invariants]
  In simply-typed languages invariants on concrete types are left implicit: as long as the implementer of library code preserves the invariants, clients must as well, because they may only use operations provided by the interface.
  In dependent type theory types may encode the invariants explicitly; for example, the type $\rbtty{E}$ represents only valid \emph{red-black} trees.
  By this property alone, all functions defined on $\rbtty{E}$ must respect the red-black invariants, even prior to the preservation of abstraction functions afforded by the abstract phase.
\end{remark}

Now, as in the previous example, we may use gluing to define the type of red-black trees representing lists,
\[ X \isdef \Glue{\Op_{\abs}(\listty{E})}{\Cl_{\abs}(\rbtty{E})}{\Cl_{\abs}(\etaOp \circ \Impl{elements})}, \]
where $\Impl{elements} : \rbtty{E} \to \listty{E}$ computes the in-order traversal of a red-black tree.
As before, we can define an operation on $X$ by giving an abstract specification and a concrete implementation that cohere up to this $\Impl{elements}$ abstraction function.

\begin{remark}\label{rem:phase-restrictions}
  Gluing with respect to the abstract phase, we commit to an abstract model of red-black trees, $\Op_{\abs}(\listty{E})$.
  This enforces that all functions out of $X$ must be portable to the abstract phase as operations on $\listty{E}$, analogous to Kripke's possible world semantics \cite{kripke>1963} where the abstract phase is a ``future world''.
  For example, it is impossible to write a program
  $\Impl{isRed} : X \to \boolty$
  that computes whether the root node of a tree is red, because there is not enough data in the list representation to determine whether the root of a corresponding tree (that is, a tree in the preimage of $\Impl{elements}$) was red, indicated by the dotted arrow in the left diagram:
  \begin{center}
  \begin{minipage}{0.5\linewidth}
\[\begin{tikzcd}
	{\rbtty{E}} & \boolty \\
	{\listty{E}} & \boolty
	\arrow["{\TOP{\Impl{isRed}}}", from=1-1, to=1-2]
	\arrow["{\Impl{elements}}"', from=1-1, to=2-1]
	\arrow[r,-,double equal sign distance,double, from=1-2, to=2-2]
	\arrow[dotted, from=2-1, to=2-2]
\end{tikzcd}\]
  \end{minipage}%
  \begin{minipage}{0.5\linewidth}
\[\begin{tikzcd}
	{\rbtty{E}} & \boolty \\
	{\listty{E}} & 1
	\arrow["{\TOP{\Impl{isRed}}}", from=1-1, to=1-2]
	\arrow["{\Impl{elements}}"', from=1-1, to=2-1]
	\arrow[from=1-2, to=2-2]
	\arrow[from=2-1, to=2-2]
\end{tikzcd}\]
  \end{minipage}%
  \end{center}
  However, such a function could be successfully implemented with type
  $\Impl{isRed} : X \to \Cl_{\abs} \boolty$ as in the right diagram,
  computing this data \emph{concretely} without revealing it abstractly.
  This is because the abstract aspect of the concrete $\Cl_{\abs} \boolty$ is trivial, so the abstract specification
  is trivial.
\end{remark}

\subsection{Abstraction Functions Implicitly, via Phased Quotients}\label{sec:abstraction:quotient}

In the previous subsection we used gluing to construct types with a desired abstract aspect, and maps between such types consisted of a concrete implementation, an abstract specification, and a commutative square, adapted to the phased setting.
Although this approach is technically always applicable by \cref{thm:fracture}, storing both representations side-by-side is not always ergonomic: we are forced to write the specification and implementation as separate pieces of code, even though they may share substantial content.
Moreover, verification of their coherence is \emph{global}: the abstraction function converts the concrete type to the abstract type, and we must show that operations (typically defined by recursion) cohere (typically proved by induction).

\begin{remark}
  A similar issue of ergonomics occurs in module calculi~\cite{harper-mitchell-moggi>1989}: it is inconvenient for a programmer to ``physically'' separate types and terms in modules, referred to as a \emph{phase separation}.
  Instead, it is preferable to intermix types and terms and isolate the static type components, using a \emph{phase distinction} as in the module calculus of \citet{sterling-harper>2021}.
\end{remark}

Through the use of inductive types with an in-phase quotient, we may make verification \emph{local}, indicating which components of a concrete type should be collapsed to achieve the intended abstract type.%
\footnote{Although we are working in a univalent setting, the types we consider will all be set-truncated.
Thus, when we define a higher inductive type (in this work, only quotient types), we omit explicit set-truncation for brevity.}
This blends the abstract type into the concrete type more artfully, leaving the abstraction function implicit due to the use of a synthetic phase distinction.

\subsubsection{Batched Queue}\label{sec:abstraction:quotient:queue}

To locally quotient the pair-of-lists type used to implement batched queues, we allow elements to move freely between the end of the inbox list $l_1$ to the end of the outbox list $l_2$, depicted in \cref{fig:batch-quotient}.
This technique is a phased refinement of the batched queue type of \citet[\S 4.2]{angiuli-cavallo-mortberg-zeuner>2021}, providing the $\Label{tilt}$ quotient only in the abstract phase: although it makes sense to shift elements between the inbox and outbox when reasoning abstractly about correctness, such movement concretely would defeat the purpose of having an efficient batched queue implementation.
In this manner we retain the reasoning benefits of \citet{angiuli-cavallo-mortberg-zeuner>2021} in the abstract phase while gaining the ability to extract quotient-free concrete code using the out-of-phase \cref{mod:false}.
\begin{figure}
  \iblock{
    \mhang{\kw{data}~\batchty{(E : \tpv)} : \tpv~\kw{where}}{
      \mrow{\Label{inj} : \listty{E} \times \listty{E} \to \batchty{E}}
      \mrow{\Label{tilt} : \abs \to (e : E)~(l_1~l_2 : \listty{E}) \to \Label{inj}~(\catlist{l_1}{\singex{e}}, l_2) = \Label{inj}~(l_1, \catlist{l_2}{\singex{e}})}
    }
  }
  \caption{Type representing batched queues, instrumented with an abstract-phase quotient to allow data to flow implicitly from the inbox list to the outbox list.}\label{fig:batch-quotient}
  \Description{Purely textual code figure.}
\end{figure}

Implementing operations on this type is straightforward.
For example, we may define the empty queue and the enqueue operation as follows (omitting the dequeue operation for brevity):
\begin{align*}
  &\Impl{empty} : \batchty{E} &
  &\Impl{enqueue} : E \to \batchty{E} \to \batchty{E} \\
  &\Impl{empty} \isdef \Label{inj}~(\nilex, \nilex) &
  &\Impl{enqueue}~e~(\Label{inj}~(l_1, l_2)) \isdef \Label{inj}~(\consex{e}{l_1}, l_2) \\
  &&
  &\Impl{enqueue}~e~(\Label{tilt}~\b~e'~l_1~l_2) \isdef \Label{tilt}~\b~e'~(\consex{e}{l_1})~l_2
\end{align*}
These implementations precisely follow the concrete implementations $\TOP{\Impl{empty}}$ and $\TOP{\Impl{enqueue}}$ from \cref{sec:abstraction:gluing:queue}, with the minor addition of the $\Label{tilt}$ case in $\Impl{enqueue}$ to justify that adding $e$ to the beginning of $l_1$ respects the abstract-phase quotient.
However, they implicitly incorporate both the abstract specifications $\ABS{\Impl{empty}}$ and $\ABS{\Impl{enqueue}}$ and the proof of coherence.
The lifted function \[ \Impl{revAppend} : \batchty{E} \to \listty{E} \] is an abstract equivalence,
witnessing the fact that $\batchty{E}$ has the intended abstract semantics.

\subsubsection{Red-Black Tree}\label{sec:abstraction:quotient:rbt}

With additional phase-sensitivity, this technique may be applied to the construction of red-black trees, as well.
Within a red-black tree (\cref{fig:rbt}), alongside the abstractly-relevant sequence of elements is (1) a tree shape and (2) a red/black node coloring, both stored for efficiency of the concrete implementation.
Thus, we may use a phased quotient to
\begin{enumerate}
  \item add associativity and identity laws to identify trees up to rotation and
  \item add a recoloring law to identify the $\Label{black}$ and $\Label{red}$ node constructors,
\end{enumerate}
deleting the concrete data in the abstract phase.
However, with indices of type $\colorty$ and $\nat$, na\"ively grafting on these laws is not even well-typed:
performing arbitrary tree rotations and recolorings need not lead to another invariant-satisfying red-black tree!
To evade the red-black invariants, the key maneuver is the placement of the color and black-height invariants under the concrete modality $\Cl_{\abs}$: that way, in the abstract phase, we are no longer obliged to maintain the red-black invariants.
This revised quotient inductive type is shown in \cref{fig:rbt-quotient}.
\begin{figure}
  \iblock{
    \mhang{\kw{data}~\Name{IRBTree}~ (E : \tpv) : (c : \Cl_{\abs}\colorty) \to (n : \Cl_{\abs}\nat) \to \tpv ~\kw{where}}{
      \mrow{\Label{empty} : \irbtty{(\etaCl \black)}{(\etaCl \zero)}{E}}
      \mrow{\Label{leaf} : E \to \irbtty{(\etaCl \black)}{(\etaCl \zero)}{E}}
      \mrow{\Label{red} : \irbtty{(\etaCl \black)}{n}{E} \to \irbtty{(\etaCl \black)}{n}{E} \to \irbtty{(\etaCl \red)}{n}{E}}
      \mrow{\Label{black} : \irbtty{c_1}{n}{E} \to \irbtty{c_2}{n}{E} \to \irbtty{(\etaCl \black)}{((\Cl_{\abs}\suc*{})~n)}{E}}
      \mrow{\Label{recolor} : \abs \to (t_1~t_2 : \irbtty{\starCl}{\starCl}{E}) \to \Label{red}~t_1~t_2 = \Label{black}~t_1~t_2}
      \mrow{\Label{id}^\Label{l} : \abs \to (t : \irbtty{\starCl}{\starCl}{E}) \to \Label{black}~t~\Label{empty} = t}
      \mrow{\Label{id}^\Label{r} : \abs \to (t : \irbtty{\starCl}{\starCl}{E}) \to \Label{black}~\Label{empty}~t = t}
      \mrow{\Label{assoc} : \abs \to (t_1~t_2~t_3 : \irbtty{\starCl}{\starCl}{E}) \to \Label{black}~(\Label{black}~t_1~t_2)~t_3 = \Label{black}~t_1~(\Label{black}~t_2~t_3)}
    }
    \mrow{}
    \mrow{\rbtty{E} \isdef \sum_{c : \Cl_{\abs}\colorty} \sum_{n : \Cl_{\abs}\nat} \irbtty{c}{n}{E}}
  }
  \caption{Type representing invariant-preserving red-black trees, instrumented with the concrete modality and abstract quotients to  annihilate red-black coloring and tree shape in the abstract phase.}\label{fig:rbt-quotient}
  \Description{Purely textual code figure.}
\end{figure}
As in the batched queue example, we must now equip programs with proofs that they respect the local perturbations allowed by the phased quotient.
For example, the append operation that combines two red-black trees in an order-preserving and invariant-maintaining manner~\citep{blelloch-ferizovic-sun>2016,blelloch-ferizovic-sun>2022,li-grodin-harper>2023}
\[ \Impl{append} : \rbtty{E} \to \rbtty{E} \to \rbtty{E} \]
just performs rotations and recolorings, which respects the given quotients.

\begin{remark}[Smart Constructors]\label{rem:smart-constructor}
  For the verification of $\Impl{append}$, the essential lemma is that
  \[ \eqOp{\Impl{append}}{\Label{black}}. \]
  In other words: from an abstract perspective, $\Impl{append}$ is simply the $\Label{black}$ node constructor.
  For this reason, we justify the terminology that $\Impl{append}$ is a smart constructor, informally defined to be a constructor that performs some additional computation for the sake of concrete efficiency only.
  We may treat this observation as a formal definition using the abstract phase: a \emph{smart constructor} is a program that is abstractly equivalent to a constructor.
\end{remark}

\begin{remark}\label{rem:semantics-quotient}
  Because the invariants and quotients are both relative to the abstract phase, we may extract different aspects of the type using various semantics from \cref{sec:semantics}.
  For example,
  \begin{enumerate}
    \item
      in \cref{mod:true} the indices are erased, and the quotients always apply, recovering the mathematical free monoid (equivalent to $\listty{E}$) and specification code; and
    \item
      in \cref{mod:false} the concrete modality on the indices becomes a no-op, and the quotient equations are erased, recovering the standard red-black tree type (\cref{fig:rbt}) and code.
  \end{enumerate}
  The phase mediates between these semantics: code must respect the possibility of collapsing to lists under the phase, but this may uniformly be deleted in the semantics to recover the true code.
\end{remark}
\section{Phased Interfaces for Modularity}\label{sec:interface}

Through the use of the abstract phase, every type consists of both concrete and abstract aspects. When a client uses an abstract data type, it may only be concerned with the abstract aspects when verifying its own abstract properties, remaining agnostic to their concrete implementation. In this sense the abstract aspect of a type acts as an \emph{interface} on which a client may rely.

In a typical module system maintaining this separation is crucial: the client must not depend on implementation details of the library in order to preserve abstraction and modularity.
This section describes how to restrict the abstract aspect of types and programs to give modular interfaces representing algorithms and data structures, and formalize the idea of modularity through a critical \emph{noninterference} property of the phase distinction.

\subsection{Abstract Data Types}\label{sec:interface:adt}

Traditionally, an abstract data type is informally understood to be a mathematical model of a type and some operations on the type.
\begin{myquote}{liskov-zilles>1974}
  An \emph{abstract data type} defines a class of abstract objects which is completely characterized by the operations available on those objects\dots
  When a programmer makes use of an abstract data object, he is concerned only with the behavior which that object exhibits but not with any details of how that behavior is achieved by means of an implementation.
\end{myquote}
To describe an abstract data type, it is essential to describe the behavior a client can expect from its data structure implementations, including the possible abstract states a client can envision the type to contain and the abstract mathematical behavior of each given operation in terms of these states.

Continuing our running example, let us attempt to define an abstract data type of queues.
We recall from \cref{sec:abstraction:gluing:queue} the queue interface type, which includes a type $X$ and empty, enqueue, and dequeue operations of the usual types,
\[ \Name{PreQueue} \isdef \sum_{X : \tpv} (\Label{empty} : X) \times (\Label{enqueue} : E \to X \to X) \times (\Label{dequeue} : X \to E \times X). \]
A reader might intuit from the included label names---$\Label{empty}$, $\Label{enqueue}$, and $\Label{dequeue}$---that this interface describes an abstract data type that contains elements of type $E$.
Which abstract data type is being described, though---stacks, queues, sets, priority queues, or something else?
These labels of course carry no technical meaning, which allows implementations of $\Name{PreQueue}$ that behave nothing like queues, and it provides insufficient information to clients for verification.
The abstract guarantees made about inhabitants can be restricted using an \emph{abstract specification type}.

\begin{definition}\label{def:specification-type}
  The \emph{$\varphi$-phase specification type} for specification $\BEH{x} : \Op_{\varphi} X$ is the type
  \[ \SPEC[\varphi]{X}{\BEH{x}(\b)} \isdef \sum_{x : X} \Op_{\varphi}(x = \BEH{x}(\b)), \]
  describing all the inhabitants $x : X$ that cohere with $\BEH{x}$ in the phase $\varphi$.%
  \footnote{We use a notation inspired by extension types~\citep{riehl-shulman>2017,sterling-harper>2021}. However, we use typal equality, because the proof that a program matches its specification need not be definitional.}
\end{definition}

Using an abstract specification type with $\Impl{listQueue} \isdef (\listty{E}, \ABS{empty}, \ABS{enqueue}, \ABS{dequeue})$,
we may build on $\Name{PreQueue}$ to define abstract data type of queues,
\[ \Name{Queue} \isdef \SPEC[\abs]{\Name{PreQueue}}{\Impl{listQueue}}, \]
restricting to inhabitants that are in agreement with the list-level operations given in \cref{sec:abstraction:gluing:queue} under the abstract phase. Immediately, the specification $\Impl{listQueue}$ is an inhabitant of $\Name{Queue}$.
The batched representations of \cref{sec:abstraction:gluing:queue,sec:abstraction:quotient:queue} also implement this specification, as their types and terms abstractly match $\Impl{listQueue}$.

\subsection{Noninterference and Modularity}\label{sec:interface:noninterference}

It has long been understood that the behavior of code dependent on abstract data types should only depend on the behavior guaranteed by the ADT.
\begin{myquote}{hoare>1972}
  The correctness of the final concrete program depends only on the correctness of the original abstract program.
\end{myquote}
Moreover, modularity follows from a noninterference principle, which guarantees that abstractly, concrete implementation details have no impact on verification.
\begin{myquote}{dijkstra>1965}
  [Modularity] relies on something less outspoken, viz. on what I should like to call "The principle of non-interference": \dots
  the correct working of the whole can be established by taking, of the parts, into account their exterior specification only, and not the particulars of their interior construction.
\end{myquote}
Using the synthetic phase distinction between concrete/interior and abstract/exterior data, these informal principles are realized as theorems.
First, we observe that the abstract specification type is a \emph{concrete} type, having room for concrete variation but completely fixing the abstract specification.

\begin{lemma}\label{lem:ext-is-closed-modal}
  The abstract specification type $\SPEC[\abs]{X}{\BEH{x}(\b)}$ is concrete.
\end{lemma}

Next, we present a dependent variation of the \emph{noninterference} principle for phases by~\citet[Theorem 2.5]{niu-sterling-grodin-harper>2022} that guarantees that concrete implementation details have no effect on the abstract behavior of downstream client code.

\begin{theorem}[Noninterference]\label{thm:noninterference}
  Let $X$ be a concrete type, and let $Y : X \to \tpv$ be a family of types.
  Then, in the phase, the function space $\dto{x}{X} Y(x)$ is equivalent to $\sum_{x : X} Y(x)$.
\end{theorem}
\begin{proof}
  This follows from \cref{lem:closed-modal}, which says that in the phase, $X$ is contractible.
\end{proof}

Treating $X$ as the type given to a library and $Y$ as the type of downstream code, we have that downstream code dependent on the library of type $\dto{x}{X}{Y(x)}$ has no abstractly-visible preference for any particular $X$, as it is equivalent to any choice of library $x : X$ in the abstract phase.
In other words noninterference guarantees that private, concrete data does not affect public, abstract data, except insofar as the concrete data is partially revealed in the abstract phase.
Therefore, we may freely substitute one library implementation for another without changing the abstract behavior of the client.
This property may be formalized as a corollary about \emph{modularity}.
\begin{corollary}[Modularity]\label{cor:modularity}
  Let $X$ be a concrete type, and let $Y$ be an arbitrary type.
  Then, for all $f : X \to Y$ and $x, x' : X$, we have that $\eqOp{f(x)}{f(x')}$.
\end{corollary}
In other words: for an implementation $f$ of $Y$ depending on a library of type $X$, we may freely swap any $x : X$ for any $x' : X$ and guarantee identical abstract behavior.
We now capitalize on this principle to verify client programs of the queue abstract data type.

\begin{samepage}
\begin{example}
  Consider the following program demonstrating a simple usage pattern for a queue, enqueueing an element to an empty queue and immediately dequeueing:
  \iblock{
    \mrow{\Impl{demo} : \Name{Queue} \to E \to E}
    \mrow{\Impl{demo}~q~e \isdef (q\proj{dequeue}~(q\proj{enqueue}~e~q\proj{empty}))\proj{proj}_1}
  }
  To verify the abstract behavior of $\Impl{demo}$, we show that $\eqOp{\Impl{demo}~q~e}{e}$ for all queue implementations $q : \Name{Queue}$ and elements $e : E$.
  By \cref{cor:modularity} (since $\Name{Queue}$ is concrete), we may choose a convenient implementation for verification, such as $\Impl{listQueue}$, without loss of generality.
  Then, the abstract equivalence between $\Impl{demo}~\Impl{listQueue}~e$ and $e$ holds judgmentally.
\end{example}
\end{samepage}

\begin{figure}
  \begin{minipage}{0.48\linewidth}
    \iblock{
      \mrow{\Impl{fromList} : \dto{q}{\Name{Queue}} \listty{E} \to q.X}
      \mrow{\Impl{fromList}~q~\nilex \isdef q\proj{empty}}
      \mrow{\Impl{fromList}~q~(\consex{e}{es}) \isdef q\proj{enqueue}~e~(\Impl{fromList}~q~es)}
    }
  \end{minipage}%
  \begin{minipage}{0.52\linewidth}
    \iblock{
      \mrow{\Impl{toList} : \dto{q}{\Name{Queue}} \nat \to q.X \to \listty{E}}
      \mrow{\Impl{toList}~q~\zero~x \isdef \nilex}
      \mhang{\Impl{toList}~q~(\suc{k})~x \isdef}{
        \mrow{\letex{q\proj{dequeue}~x}{(e, x')}}
        \mrow{\consex{e}{\Impl{toList}~q~k~x'}}
      }
    }
  \end{minipage}

  \bigskip
  \begin{minipage}{0.48\linewidth}
    \iblock{
      \mrow{\Impl{qreverse} : \Name{Queue} \to \listty{E} \to \listty{E}}
      \mrow{\Impl{qreverse}~q~l \isdef \Impl{toList}~q~\len{l}~(\Impl{fromList}~q~l)}
    }
  \end{minipage}%
  \begin{minipage}{0.52\linewidth}
    \iblock{
      \mrow{\Impl{reverse} : \listty{E} \to \listty{E}}
      \mrow{\Impl{reverse}~\nilex \isdef \nilex}
      \mrow{\Impl{reverse}~(\consex{e}{es}) \isdef \catlist{\Impl{reverse}~es}{\singex{e}}}
    }
  \end{minipage}
  \caption{List reverse implemented using a queue, $\Impl{qreverse}$, and a direct list reversal function, $\Impl{reverse}$.}\label{fig:queue-rev}
  \Description{Purely textual code figure.}
\end{figure}

\begin{example}
  Another simple usage of queues is reversing a list by enqueueing its elements and dequeueing them in reverse order, implemented in $\Impl{qreverse}$ in \cref{fig:queue-rev}.
  We may verify that the abstract behavior of $\Impl{qreverse}$ matches the specification, $\Impl{reverse}$: we show that $\eqOp{\Impl{qreverse}~q}{\Impl{reverse}}$ for all queue implementations $q : \Name{Queue}$.
  Again by \cref{cor:modularity}, we may choose a convenient implementation for verification.
  Choosing $\Impl{listQueue}$ for $q$ in $\Impl{fromList}$ and $\Impl{toList}$, it is immediate that $\Impl{fromList}$ is exactly the $\Impl{reverse}$ function, and $\Impl{toList}$ is exactly the identity function. Thus, the abstract behavior of $\Impl{qreverse}~q$ is equivalent to $\Impl{reverse}$, for all possible implementations $q : \Name{Queue}$.
\end{example}

\begin{remark}[Synthetic parametricity]\label{rem:relational-parametricity}
  In their presentation of batched queues \citet[\S 4.1]{sterling-harper>2021} similarly provide a conjoined implementation of a list queue (here, $\Impl{listQueue}$) alongside a batched queue (here, $\Impl{batchedQueue}$) connected by a (functional) relation, $\Impl{revAppend}$.
  Using a pair of \emph{symmetric} (``left'' and ``right'') phases, either the list queue or the batched queue may be isolated by entering the appropriate phase.
  Instead, we emphasize here the functional, \emph{asymmetric} nature of the relation (given as $\chi$): we only allow a coercion of our conjoined implementation to the privileged specification $\Impl{listQueue}$, chosen as the canonical meaning of ``queue'', via the abstract phase.
  We recover a theorem analogous to their representation independence result \cite[Theorem 4.1]{sterling-harper>2021} via modularity (\cref{cor:modularity}): for all functions $f : \Name{Queue} \to \boolty$, choosing $x$ and $x'$ to be $\Impl{listQueue}$ and $\Impl{batchedQueue}$ respectively, we have that $\eqOp{f(\Impl{listQueue})}{f(\Impl{batchedQueue})}$.
\end{remark}

\begin{remark}[Security]
  Modularity and the abstract phase may be viewed through the lens of security, where $\abs$ corresponds to a public, low-security environment and $\top$ corresponds to a private, high-security environment~\citep{sterling-harper>2022}.
  By default, we operate in the private environment, capable of writing secret implementation details pertaining to cost and clever implementation.
  When we switch to the public environment, though, implementation details are redacted.
  It is in this sense that the abstract phase guarantees a notion of abstraction and modularity: private data is guaranteed to be redacted in the public phase and may therefore be swapped with any other private data at will with no effect.
\end{remark}

\subsection{Phased Universal Properties}

Modular reasoning via noninterference is available for any concrete type, not just those syntactically constructed as specification types.
Commonly, abstract data types classify free/inductive algebraic structures; for example, finite ordered sequences are classified by the free monoid, finite multisets are classified by the free commutative monoid, and finite sets are classified by the free semilattice.
In this section, using red-black trees implementing finite ordered sequences as an example, we form a concrete type that classifies data structures implementing a free algebraic structure.

\subsubsection{Abstract Properties}\label{sec:interface:adt:properties}

Sequences are typically viewed as a monoid, but this statement has a delicate interaction with abstraction.
For example, the $\Impl{append}$ operation on red-black trees (\cref{sec:abstraction:quotient:rbt}) is not associative: appending trees in different orders will typically result in distinct shapes and colorings of red-black trees.
However, it is \emph{abstractly} associative: in the abstract phase, it is the case by \cref{rem:smart-constructor} that $\Impl{append}$ is simply the $\Label{black}$ constructor, and the $\Label{assoc}$ law allows rotation of black node constructors.
This pattern is generally true of data structures: the imagined properties satisfied are only realized \emph{abstractly}, as concrete data that is not reflected abstractly need not respect the given properties.
Although a monoid will consist of the same ``stuff'' and ``structure'' as usual, the ``properties'' (written as $\Name{IsMonoid}$) are only stated abstractly:
\[ \OMonoid \isdef \sum_{X : \tpv} (\Label{empty} : X) \times (\Label{append} : X \to X \to X) \times \Op_{\abs}(\Name{IsMonoid}(\Label{empty}, \Label{append})) \]
Abstractly, this definition coheres with the standard mathematical definition of a monoid: entering the abstract phase recovers the classical notion.

\subsubsection{Abstract Universal Properties}\label{sec:interface:adt:up}
To develop sequences as the free monoid, abstract monoids are equipped with generators of type $E : \tpv$, serving to create a singleton sequence:
\[ \OMonoidOn{E} \isdef \sum_{M : \OMonoid} (\Label{singleton} : E \to M.X) \]
The universal property of the free monoid is a function $\Label{mapreduce}$ that maps into every other abstract monoid on $E$.
Typically, this map would be a homomorphism, but just as the append operation for red-black trees is not truly associative, $\Label{mapreduce}$ need not truly preserve the $\OMonoidOn{E}$ structure.
Instead, we will ask that $\Label{mapreduce}$ only \emph{abstractly} preserve structure.
\begin{definition}
  Let $M, M' : \OMonoidOn{E}$.
  An \emph{abstract homomorphism} from $M$ to $M'$ consists of a function $f : M.X \to M'.X$ that abstractly preserves the operations: %
  \begin{align*}
    \eqOp*{f (M\proj{empty})}{M'\proj{empty}} \\
    \eqOp*{f (M\proj{append}~x_1~x_2)}{M'\proj{append}~(f~x_1)~(f~x_2)} \\
    \eqOp*{f (M\proj{singleton}~e)}{M'\proj{singleton}~e}
  \end{align*}
  Write the type of abstract homomorphisms from $M$ to $M'$ as $\Hom[\circ]{M}{M'}$.
\end{definition}

Using abstract homomorphisms, the concrete type of sequences may be defined as the abstractly-initial $\OMonoidOn{E}$:
\[ \Sequence{E} \isdef \sum_{M : (\OMonoidOn{E})} \prod_{M' : (\OMonoidOn{E})} \sum_{f : \Hom[\circ]{M}{M'}} \prod_{f' : \Hom[\circ]{M}{M'}} \eqOp{f}{f'}. \]
In other words a sequence consists of $M : \OMonoidOn{E}$ alongside a universal $\Label{mapreduce}$ function (taking $M'$ as input).
In the abstract phase this type restricts to the usual mathematical definition of a free monoid on $E$, which is unique, causing the type $\Sequence{E}$ to be concrete by \cref{lem:closed-modal}.
This reconciles within a phased language the idea that sequences have many implementations: although the free monoid is abstractly unique, concretely there is room for variation.

\subsubsection{Abstract Refinements}

Although the universal property of sequences provides the facility to implement any algorithm, there is no guarantee of efficiency.
The type $\Sequence{E}$ may be refined with extra data, exporting additional operations with known behavior but better efficiency, using the fact that concrete types are closed under dependent sum~\citep[Example 1.8]{rijke-shulman-spitters>2020}.

\begin{lemma}\label{lem:sigma-closed-modal}
  If $X : \tpv$ and $Y : X \to \tpv$ are concrete, then $\sum_{x : X} Y(x)$ is concrete.
\end{lemma}

Letting $X \isdef \Sequence{E}$, which is concrete by \cref{lem:ext-is-closed-modal}, we may define concrete type families $Y$ that equip a sequence with additional data.

\begin{example}
  For any sequence implementation $M : \Sequence{E}$, the length of a given sequence can be computed using $\Label{mapreduce}$, the second projection of $M$:
  \[ \Impl{length} \isdef M\proj{mapreduce}~(\nat, 0, +, +\text{-}0\text{-}\Impl{isMonoid}, \Impl{const}~1) \]
  However, for many implementations of the sequence abstraction (including red-black trees), this operation will take linear time.
  We may refine sequences with an additional primitive operation that must behaviorally cohere with the above implementation:
  \[ \Name{SequenceExt}(E) \isdef \sum_{M : \Sequence{E}} (\Label{length} : \SPEC[\abs]{M.X \to \nat}{\Impl{length}}) \]
  Even though this $\Label{length}$ function must abstractly cohere with $\Impl{length}$ (and could always be implemented as exactly $\Impl{length}$ above), it may also be implemented using a more efficient method, such as storing the length alongside the sequence for constant-time computation.
  Notice that by \cref{lem:ext-is-closed-modal}, the type $\SPEC[\abs]{M.X \to \nat}{\Impl{length}}$ is concrete for every $M$.
  Therefore, by \cref{lem:sigma-closed-modal}, $\Name{SequenceExt}(E)$ is concrete as well.
\end{example}

This strategy of exporting abstractly-redundant information in an abstract data type is pervasive, as the algorithms able to be implemented directly via the induction principle are rarely the most efficient.
For example, the queue abstract data type can be thought of as a refinement of lists with optimized operations for appending elements to the end and removing elements from the beginning, and priority queues can be thought of as refining finite multisets with an optimized operation for removing the least element according to some ordering.
\section{Cost Refinements}\label{sec:cost}

Thus far, the abstract phase distinction has separated concrete implementation details, used for improved efficiency, from abstract specifications.
Now, following Calf~\citep{niu-sterling-grodin-harper>2022}, cost can be reified in programs as an effect.
Accordingly, the abstract phase $\abs$ bifurcates, providing the compatible ability to erase cost information from programs.

\subsection{Dependent Type Theory with Cost}\label{sec:cost:cost}

To support effectful programs in dependent type theory, we transition to Calf~\citep{niu-sterling-grodin-harper>2022}, a dependent type theory that distinguishes between running computations and inert values in both terms and types.
Calf is a dependent variation of call-by-push-value~\cite{levy>2003}, inspired by \citet{ahman-ghani-plotkin>2016}, \citet{vakar>thesis}, and \citet{pedrot-tabareau>2019}.
\begin{align*}
  \text{Val.} \quad X,Y,Z &\coloncolonequals \U{A} \mid \cdots \\
  \text{Comp.} \quad A,B,C &\coloncolonequals \F{X} \mid \dpto{x}{X} A(x)
\end{align*}
The value types extend dependent type theory (including the universe of value types $\tpv$) with the suspensions $\U{A}$ of call-by-push-value, and the computation types include effectful computations of value types $\F{X}$ and dependent products (functions) from a value type $X$ to a family of computation types $A(x)$.
We use common notations for programs involving these types, such as pattern matching on $\ret{x}$ (the introduction form for the $\F{X}$ type) as the elimination form for $\F{X}$.
Following Calf~\citep{niu-sterling-grodin-harper>2022}, we also use the ``less bureaucratic'' form of call-by-push-value in which the introduction and elimination forms for the $\U{A}$ type are left implicit, and we silently convert between the equivalent types $\U{\dpto{x}{X} A(x)}$ and $\dto{x}{X} \U{A(x)}$.

In Calf cost is treated abstractly as an effect: we write $\charge[A]{c}{a}$ for a print-like effect that records $c : \costty$ units of abstract cost before running the computation $a : A$, where $\costty$ is a value type representing cost equipped with a monoid structure $(0, +)$.%
\footnote{We alter the notation $\mathsf{step}^c(-)$ from previous developments in Calf, emphasizing that $c$ is counting an abstract notion of cost, not the number of evaluation steps.
This distinction is of particular relevance here, where we use implausible-looking annotations for abstract reasoning in the abstract phase.}
\begin{mathpar}
  \inferrule
    {
      \Gamma \vdash c : \costty \\
      \Gamma \vdash a : A
    }
    {\Gamma \vdash \charge[A]{c}{a} : A}
\end{mathpar}
The cost effect must respect the monoid structure $(\costty, 0, +)$:
\begin{subequations}
  \begin{align}
    \charge{0}{a} &= a \label{eq:zero}\\
    \charge{c_1}{\charge{c_2}{a}} &= \charge{c_1 + c_2}{a} \label{eq:plus}
  \end{align}
\end{subequations}
For the examples here, we implicitly assume a monoid homomorphism $(\nat,0,+) \to (\costty,0,+)$, treating costs as numbers that are combined additively.

As is standard in call-by-push-value, computation types are semantically interpreted as algebras over a monad on $\tpv$; here, we use the writer monad, $\costty \times (-)$~\citep{levy>2003,niu-sterling-grodin-harper>2022}.

\subsection{Abstract Cost Specifications as Concrete Types}

In the presence of the cost effect the role of concrete types (such as abstract specification types) evolves: the inhabitants of a concrete type still collapse to a single representative in the abstract phase, but now, this includes cost.
Before, through the use of an abstract specification type, we could restrict attention to programs that matched a fictitious mental model; now, this mental model must also include cost annotations.

\begin{principle}\label{principle}
  Although we maintain an intended cost model on general programs, we intentionally allow arbitrary cost annotations on programs in the abstract phase, however implausible they may seem on the surface.
  This is because abstract-phase cost annotations will act as costs for reasoning---not execution---purposes, akin to the numbers used when solving cost recurrence relations traditionally.
\end{principle}

\begin{example}\label{ex:log-linear-sort}
  Recall the type $\rbtty{E}$ of red-black trees that appear as lists in the abstract phase (\cref{fig:rbt-quotient}).
  The type of $n \lg n$-cost sorting algorithms on such red-black trees is given (up to the abstract equivalence between $\rbtty{\nat}$ and $\listty{\nat}$) by the abstract specification type
  \[ \SPEC[\abs]{\U{\rbtty{\nat} \pto \F{\rbtty{\nat}}}}{\lambda l.\ \charge{\len{l}\lg\len{l}}{\ret{\ImplSpec{sort}~l}}}, \]
  where $\ImplSpec{sort} : \listty{\nat} \to \listty{\nat}$ is the unique sorting function.
  In the abstract phase we have that $\rbtty{\nat}$ is equivalent to $\listty{\nat}$, which is coherent (up to equivalence) with the specification $\lambda l.\ \charge{\len{l}\lg\len{l}}{\ret{\ImplSpec{sort}~l}}$ having type $\U{\listty{\nat} \pto \F{\listty{\nat}}}$.
  This type classifies all effectful programs on red-black trees that, when viewed in the abstract phase, have log-linear cost and behave like a sorting algorithm.
  By design, the log-linear cost need not be realistic to $\ImplSpec{sort}$, as the specification is given in the abstract phase.
\end{example}

As in \cref{sec:interface:noninterference}, using an abstract specification type has modularity benefits: when verifying client code in the abstract phase, we may freely substitute the given specification for any implementation.
By adding cost as an effect, our notion of interface relative to the abstract phase has become cost-sensitive, conveying an intended cost and behavior to clients by restricting implementations accordingly.
In doing so, though, we have lost the ability to describe properties that are cost-independent.
For example, when \cref{sec:interface} is na\"ively adapted to the effectful setting of Calf,
\begin{itemize}
  \item the $\Name{Queue}$ type now describes queues with zero cost;
  \item the $\OMonoid$ type classifies append functions that are associative including cost (excluding essentially all intended inhabitants, such as red-black trees); and
  \item the $\Name{Sequence}$ type classifies sequence implementations with zero cost, to satisfy initiality.
\end{itemize}
To recover the intended meaning of these types, describing properties that are oblivious to cost, we bifurcate the prior uses of the abstract phase into two phases: the abstract phase $\abs$ for isolating the abstract semantics of a program, and the behavioral phase $\beh$ for isolating the behavior of a program by erasing cost annotations.

\subsection{Behavioral Phase Distinction}

To isolate the behavioral semantics of programs, we add the synthetic phase distinction of Calf~\cite{niu-sterling-grodin-harper>2022} that ignores all cost profiling using the charge effect.
Specifically, the \emph{(abstract) behavioral phase}%
\footnote{We use the ``behavioral/algorithmic'' terminology in place of the ``extensional/intensional'' terminology from Calf~\cite{niu-sterling-grodin-harper>2022}, avoiding confusion with the unrelated ideas of extensional/intensional type theory and mathematical extensionality principles while emphasizing the abstraction viewpoint of this work. Moreover, both phases $\abs$ and $\beh$ can be thought of as isolating the extension of a program, including and excluding cost, respectively.}
is another proposition $\beh : \tpv$ that, when inhabited, erases details relevant only for efficiency, leaving behind only the behavior for analysis of correctness.
Because efficiency-oriented implementation details are erased by the abstract phase $\abs$ already, we posit that the behavioral phase $\beh$ is a stronger assumption than $\abs$.

\begin{axiom}\label{ax:beh-abs}
  The behavioral phase $\beh$ implies the abstract phase $\abs$.
\end{axiom}

\begin{remark}
  Given both phases, \cref{mod:psh} can be extended to phase poset $\{ \beh \vdash \abs \vdash \top \}$.
\end{remark}

The behavioral phase induces modalities $\Op_{\beh}$ and $\Cl_{\beh}$ analogous to those derived from the abstract phase~\citep{niu-sterling-grodin-harper>2022}.

\begin{terminology}
  We refer to the modality $\Op_{\beh}$ as the \emph{behavioral modality}, which imposes the phase $\beh$ to erase efficiency-oriented details.
  A type is \emph{behavioral} when $\etaOp_X : X \to \Op_{\beh} X$ is an equivalence, and an inhabitant of a behavioral type is a \emph{behavior}.
\end{terminology}

\begin{terminology}
  We refer to the modality $\Cl_{\beh}$ as the \emph{algorithmic modality}, which marks data for erasure when cost is to be ignored.
  A type is \emph{algorithmic} when $\etaCl_X : X \to \Cl_{\beh} X$ is an equivalence, and an inhabitant of an algorithmic type is an \emph{algorithm}.
\end{terminology}
Algorithmic types will adopt many of the responsibilities of concrete types, describing classes of algorithms that all share a single behavior.
Although an algorithmic type may have many inhabitants, they must all implement the same behavior, rendering the type behaviorally trivial.

\begin{lemma}
  Every behavioral type is abstract, and every concrete type is algorithmic.
\end{lemma}

Beyond the abstract phase, the behavioral phase must erase cost annotations.
To accomplish this, we assume that the type of costs $\costty$ is algorithmic: this allows the cost effect to be deleted in the behavioral phase, supporting reasoning about correctness without involving their cost.

\begin{axiom}\label{ax:cost-algo}
  The type of costs $\costty$ is algorithmic.
\end{axiom}

\begin{samepage}
\begin{theorem}\label{thm:charge}
  If $\beh$ holds, then for all $c : \costty$, we have $\charge{c}{e} = e$.
\end{theorem}
\begin{proof}
  Suppose $\beh$ holds, and let $c : \costty$ be arbitrary.
  Then, because $\costty$ is algorithmic, it is contractible.
  Therefore, all elements of $\costty$ are equal, so $c =_\costty 0$.
  The result follows by \cref{eq:zero}.
\end{proof}
\end{samepage}

As in Calf~\citep{niu-sterling-grodin-harper>2022}, we choose to define $\costty \isdef \Cl_{\beh} \nat$ to describe numeric costs.

\subsection{Behavioral Specifications as Algorithmic Types}

In \cref{sec:interface} we used the abstract phase to describe abstract data types, thereby restricting code in the abstract phase to match a given specification.
Such descriptions formed closed-modal types, which admit noninterference and modularity principles: by containing a unique in-phase representative, closed-modal types ensure that client code remains invariant under substitution of library code when in the phase.
In the presence of cost this line of reasoning directly adapts to the behavioral phase, where an algorithmic type contains a unique behavior and therefore client behavior is unaffected by the choice of library implementation.

\begin{example}\label{ex:sort}
  In \cref{ex:log-linear-sort} we gave a concrete type that classifies all sorting algorithms on red-black trees that charge log-linear cost.
  To describe \emph{all} sorting algorithms on red-black trees, we may give the behavioral specification type (\cref{def:specification-type})
  \[ \Name{Sort} \isdef \SPEC[\beh]{\U{\rbtty{\nat} \pto \F{\rbtty{\nat}}}}{\lambda l.\ \ret{\ImplSpec{sort}~l}} \]
  that imposes the weaker requirement on inhabitants to sort lists, making no mention of cost by giving a specification in the behavioral phase.
  This type is formally algorithmic, justified because it describes an algorithmic problem---sorting a list (represented as a red-black tree)---and its inhabitants are precisely the algorithms that meet this specification.
  For example, both insertion sort and merge sort inhabit this type, as sorting algorithms.%
  \footnote{Sorting algorithms that use non-cost effects in a benign manner, such as nondeterministic and randomized quicksort~\citep{grodin-niu-sterling-harper>2024}, also inhabit the type $\Name{Sort}$.}
  Although this type is algorithmic, it is not concrete: whereas the concrete type of \cref{ex:log-linear-sort} gave an integrated cost-and-behavior specification, this type provides clients with the weaker guarantee imposed by a behavior-only specification.
\end{example}

\begin{example}
  Abstract cost specifications extend to higher-order functions when the input function is assumed to be abstractly sufficiently restricted, following \citet{grodin-niu-sterling-harper>2024}.
  For example, the ``reduce'' algorithm on red-black trees takes linear time as long as the combining function is assumed to be constant time.
  (If the combining function has more complex cost characteristics, it is difficult to provide a simple bound, as the precise layout of the red-black tree may affect the cost.)
  This nuanced property can be expressed through careful use of the abstract and behavioral phases.
  Say that a function $f : \U{E \pto E \pto \F{E}}$ is total when, behaviorally, it always returns a result:
  \[ \Name{Total}(f) \isdef \sum_{f' : E \to E \to E} \Op_{\beh} (f = (\lambda~e_1~e_2.\ \ret{f'~e_1~e_2})) \]
  Refining this definition, when $f$ is total, it can be further restricted to have constant time:
  \begin{align*}
    &\Name{ConstantTime} : \Name{Total}(f) \to \tpv \\
    &\Name{ConstantTime}(f', \b) \isdef f = (\lambda~e_1~e_2.\ \charge{1}{\ret{f'~e_1~e_2}})
  \end{align*}
  If the reduce algorithm is forced to take in a proof that its input function abstractly takes constant time, then it can guarantee linear cost itself.
  Letting the type of reduce be
  \begin{align*}
    \Name{ReduceT} \isdef
      & ~\kw{U}((f : \U{E \pto E \pto \F{E}}) \pto (e : E) \\
      & \quad \pto \Op_{\abs}\left(\sum_{(f', t) : \Name{Total}(f)} \Name{ConstantTime}(f', t) \times \Name{IsMonoid}(e, f')   \right) \\
      & \quad \pto \rbtty{E} \pto \F{E}),
  \end{align*}
  the corresponding linear-time abstract specification is
  \[ \SPEC[\abs]{\Name{ReduceT}}{\lambda f. \lambda e. \lambda p. \lambda l. \ \charge{\len{l}}{\ret{\ImplSpec{fold}~f'~e~l}}}, \]
  where $\ImplSpec{fold} : (E \to E \to E) \to E \to \listty{E} \to E$ is the specification implemented as ``fold'' on lists and $f'$ is the projection from the abstract assumption $p$.
  This cost assumption places the burden of proof on the client side: when applying a function of this type, a client must prove that their given $f$ has constant cost (and the inputs form a monoid, as usual).
  Note that in \cref{mod:false}, the assumption of $\abs$ is vacuous, and therefore the proof argument is trivial; thus, in the concrete semantics, the true underlying code is extracted, with no expectations about the cost of $f$.
\end{example}

The development of \cref{sec:interface} adapts from the abstract phase to the behavioral phase similarly: behavioral specification types and universal properties underlie algorithmic types, and maps out of an algorithmic type are behaviorally constant.

\subsection{Upper Bounds via the Sealing Effect}

For a true notion of cost-aware modularity, programs matching a cost-aware specification ought only to have cost \emph{upper-bounded} by the signature.
Substituting an implementation that happens to be less expensive than the specification purports should not disrupt the soundness of downstream abstract reasoning; the true cost of the client may decrease, but this will not violate the upper bound.
Moreover, tight bounds are often inexpressible given only the information available abstractly: commonly, the exact cost of an implementation depends on concrete details that are hidden abstractly, and only an upper bound for the true cost can be computed given only abstract information.
For example, as is typical in tree algorithms, the precise cost of the red-black tree append algorithm depends on the coloring of the tree, which is abstractly inaccessible (\cref{rem:phase-restrictions}); however, an upper bound may be provided in terms of the number of nodes in the trees, which is computable from the abstract representation~\cite[Theorem 3.5]{li-grodin-harper>2023}.
To account for upper bounds, we pass to Decalf---a refinement of Calf with program inequality---and add an effect that allows the cost of programs to be weakened to an upper bound in the abstract phase.

\subsubsection{Program Inequality in Decalf}

Building on Decalf~\citep{grodin-niu-sterling-harper>2024}, we include a judgmental notion of \emph{inequality} between programs.%
\footnote{The semantics of Decalf may be straightforwardly adapted to the univalent setting~\citep{sterling>decalf-semantics}.}
Just as program equality compares the cost and behavior of programs, so too does inequality: for programs $a,a' : A$, it is the case that $a \le_A a'$ when the cost of $a$ is bounded by the cost of $a'$ and both $a$ and $a'$ have equal behavior.
Inequality in Decalf is similar to equality: it is reflexive and transitive (but not symmetric); it is pointwise on functions;
and all functions are monotone, by analogy with congruence of equality.
By monotonicity, inequality at the cost type $\costty$ lifts to inequality of cost-annotated programs:
\[ c \le_\costty c' \to \charge{c}{a} \le_A \charge{c'}{a} \]
Furthermore, inequality and equality coincide in the behavioral phase: because inequality verifies both cost and behavior simultaneously and cost is erased in the behavioral phase (\cref{thm:charge}), the only remaining force of inequality in the behavioral phase is equality.
\begin{axiom}\label{ax:decalf}
  If $\beh$ holds, then $a \le a'$ implies $a = a'$.
\end{axiom}
As before, we use natural numbers as the cost model, now imbued with the usual ordering $\costty \isdef \{0 < 1 < 2 < \cdots\}$; this type is algorithmic (to satisfy \cref{ax:cost-algo}) by \cref{ax:decalf}.

\begin{example}\label{ex:quadratic-sort}
  The insertion sort algorithm always makes at most $n^2$ comparisons, but the exact number depends on the input list.
  Refining \cref{ex:log-linear-sort}, the type of sorting algorithms (functions coherent with $\ImplSpec{sort}$) that have cost bounded by $n^2$ is
  \[ \sum_{f : \U{\rbtty{\nat} \pto \F{\rbtty{\nat}}}} \Op_{\abs}(f \le \lambda l.\ \charge{\len{l}^2}{\ret{\ImplSpec{sort}~l}}), \]
  classifying algorithms such as insertion sort~\citep[Example 3.4]{grodin-niu-sterling-harper>2024}.
  By \cref{ax:beh-abs,thm:charge,ax:decalf}, this type is behaviorally equivalent to the type $\Name{Sort}$ of \cref{ex:sort}.
  Note that while this type is algorithmic, it is \emph{not} concrete, as $f$ is not uniquely determined.
\end{example}

\subsubsection{Cost Sealing as an Effect}

To allow the cost of a program to be weakened in the abstract phase, we introduce a new effect, \emph{sealing}, that pretends that a program $a$ is a specification $\BEH{a}$ in the abstract phase, as long as $a$ is (abstractly) upper-bounded by $\BEH{a}$.
\begin{equation*}
  \inferrule
    {
      \Gamma \vdash a : A \\
      \Gamma, \b : \abs \vdash \BEH{a} : A \\
      \Gamma, \b : \abs \vdash a \le_A \BEH{a}
    }
    {\Gamma \vdash \seal[A]{a}{\BEH{a}}:A}
\end{equation*}
We omit the evidence that $a \le_A \BEH{a}$ in the sealing syntax for readability, but it does constitute a proof obligation.
This effect is the cost-aware analogue of sealing for modules, which also hides the information exported by a module under a less expressive signature~\citep{dreyer-crary-harper>2003}.

In the abstract phase sealing is erased, retaining only the specification; and sealing at a concrete type $A$%
\footnote{We say that a computation type $A$ is concrete when $\U{A}$ is concrete.}
is otiose, because the data is uniquely defined in the abstract phase.
\begin{mathpar}
  \inferrule
    {
      \Gamma \vdash \abs
    }
    {\Gamma \vdash \seal{a}{\BEH{a}} = \BEH{a}}

  \inferrule
    {
      \Gamma \vdash A~\text{concrete} \\
    }
    {\Gamma \vdash \seal[A]{a}{\BEH{a}} = a}
\end{mathpar}
Also present are unit and multiplication laws, analogous to \cref{eq:zero,eq:plus}.
Sealing $a$ at itself reflexively has no effect, as it reveals nothing new;
and sealing a program $a$ at a specification $\BEH{a}$, and then sealing again at a further specific $\BEH{a}'$, is equivalent to sealing $a$ at $\BEH{a}'$ originally, by transitivity.
\begin{mathpar}
  \inferrule
    {
      \Gamma \vdash a : A
    }
    {\Gamma \vdash \seal{a}{a} = a}

  \inferrule
    {
      \Gamma, \b : \abs \vdash a \le_A \BEH{a} \\
      \Gamma, \b : \abs \vdash \BEH{a} \le_A \BEH{a}'
    }
    {\Gamma \vdash \seal{\seal{a}{\BEH{a}}}{\BEH{a}'} = \seal{a}{\BEH{a}'}}
\end{mathpar}
Finally, sealing commutes with other effects: in particular, it commutes with the cost effect.
\begin{equation*}
  \seal{\charge{c}{a}}{\charge{c}{\BEH{a}}} = \charge{c}{\seal{a}{\BEH{a}}}
\end{equation*}

\subsubsection{Monadic Semantics of Sealing}

To provide a semantics for sealing, we continue to interpret computation types as algebras over a monad on $\tpv$.
Now, in addition to cost (using the writer monad, as in \cref{sec:cost:cost}), we must ensure that our monad describes sealing.
The monad for sealing, in fact, comes from an adaptation to the glue type: although the glue type maintained concrete and abstract data that cohere exactly up to a phased abstraction function, we merely wish to ensure a lax notion of coherence, where the concrete data is upper-bounded by the abstract data.
Accordingly, we define the \emph{lax glue type} as the following inequality-based restriction on the product $\ALGO{X} \times \BEH{X}$, also known as a comma object \citep{niefield>1981}:
\begin{center}
  \begin{minipage}{0.6\linewidth}
  \begin{align*}
    & \mathsf{Glue}^\le : \left( \sum_{\ALGO{X} : \ALGO{\tpv}} \sum_{\BEH{X} : \BEH{\tpv}} (\ALGO{X} \to \Cl_{\abs} \BEH{X}) \right) \to \tpv \\
    & \Glue*{\BEH{X}}{\ALGO{X}}{\chi} \isdef \sum_{(\ALGO{x}, \BEH{x}) : \ALGO{X} \times \BEH{X}} \chi \ALGO{x} \le_{\Cl_{\abs} \BEH{X}} \etaCl_{\BEH{X}} \BEH{x}
  \end{align*}
  \end{minipage}%
  \begin{minipage}{0.4\linewidth}
    \[\begin{tikzcd} \Glue*{\BEH{X}}{\ALGO{X}}{\chi} & \ALGO{X} \\
      \BEH{X} & {\Cl_{\abs} \BEH{X}}
      \arrow[""{name=lhs}, "\Label{proj}_2"', from=1-1, to=2-1]
      \arrow["\Label{proj}_1", from=1-1, to=1-2]
      \arrow[""{name=rhs}, "{\chi}", from=1-2, to=2-2]
      \arrow["{\etaCl_{\BEH{X}}}"', from=2-1, to=2-2]
      \arrow["\lrcorner"{anchor=center, pos=0.125}, draw=none, from=1-1, to=2-2]
      \arrow["\ge"{description}, draw=none, from=rhs, to=lhs]
    \end{tikzcd}\]
  \end{minipage}%
\end{center}
Combining this construction with fracture, we derive the \emph{sealing monad}.
\begin{definition}[Sealing Monad]
  Let the sealing monad $\SealM$ be defined as lax gluing after fracture:
  \begin{align*}
    &\SealM : \tpv \to \tpv \\
    &\SealM X \isdef \mathsf{Glue}^\le(\Frac{X}) = \sum_{(\ALGO{x}, \BEH{x}) : \Cl_{\abs} X \times \Op_{\abs} X} (\Cl_{\abs} \etaOp) \ALGO{x} \le_{\Cl_{\abs} \Op_{\abs} X} \etaCl \BEH{x}
  \end{align*}
  The unit (return) and multiplication (join) for the monad are defined as follows, implicitly using reflexivity and transitivity of $\le$, respectively:
\begin{align*}
  &\eta^{\SealM}_X : X \to \SealM X               &&\mu^{\SealM}_X : \SealM (\SealM X) \to \SealM X \\
  &\eta^{\SealM}_X~x \isdef (\etaCl x, \etaOp x)  &&\mu^{\SealM}_X~(\ALGO{s}, \BEH{s}) \isdef (\mu^{\bullet}_X(\Cl_{\abs}(\Label{proj}_1)\ALGO{s}), \mu^{\circ}_X(\Op_{\abs}(\Label{proj}_2)\BEH{s})).
\end{align*}
\end{definition}
In \cref{thm:fracture} it is shown that $\mathsf{Glue}(\Frac{X})$ is equivalent to $X$.
Although it is \emph{not} the case that $\SealM X$ is equivalent to $X$ in general, this does hold abstractly; since sealing only serves to relate abstract specifications to concrete implementations, which are erased in the abstract phase, $\eta^{\SealM}$ is abstractly an equivalence, and $\Op_{\abs} (\SealM X) = \Op_{\abs} X$.
However, while \cref{thm:fracture} says that $\Cl_{\abs} (\mathsf{Glue}(\Frac{X})) = \Cl_{\abs} X$, it is \emph{not} the case that $\Cl_{\abs}(\SealM X) = \Cl_{\abs} X$, as the upper-bounding specification cannot be reconstructed only from only the concrete implementation.

\begin{remark}
  In \cref{mod:true}, for the reasons described above, the sealing monad is interpreted as the identity monad.
  In \cref{mod:false}, the sealing monad is \emph{also interpreted as the identity monad}, even though this is not provable internally; this means that when extracting the underlying concrete code, sealing is ignored.
  In \cref{mod:psh}, the sealing monad has a nontrivial interpretation:
  \[ \SealM \Presheaf[\alpha]{\TOP{X}}{\ABS{X}} = \Presheaf[\Label{proj}_2]{\{(\TOP{x}, \ABS{x}) \in \TOP{X} \times \ABS{X} \mid \alpha(\TOP{x}) \le \ABS{x}\}}{\ABS{X}}. \]
  The bottom component matches \cref{mod:true}, as usual, and the top component describes pairs consisting of an implementation and an upper-bounding specification.
  The abstraction function projects out the upper-bounding specification, ignoring the true program.
\end{remark}

Applying the writer monad transformer to incorporate cost, algebras for the monad $\SealM(\costty \times (-))$ support both the cost and sealing effect.

\subsubsection{Sealing to Meet Cost Interfaces}

Using the sealing effect, we can artificially inflate the abstract cost of programs so that they meet their given cost specifications.

\begin{example}
  Using sealing, the type of at-most-$n^2$-cost sorting algorithms may be formulated as a concrete type, admitting reasoning by noninterference about client code (unlike \cref{ex:quadratic-sort}).
  Define the abstract specification type of at-most-$n^2$-cost sorting algorithms on red-black trees to be
  \[ \SPEC[\abs]{\U{\rbtty{\nat} \pto \F{\rbtty{\nat}}}}{\lambda l.\ \charge{\len{l}^2}{\ret{\ImplSpec{sort}}}}. \]
  Given $\Impl{isort} : \U{\listty{\nat} \pto \F{\listty{\nat}}}$, insertion sort may be sealed with an $n^2$-cost interface as
  \[ \seal{\Impl{isort}}{\lambda l.\ \charge{\len{l}^2}{\ret{\ImplSpec{sort}}}} \]
  to arise as an inhabitant of this type.
  This use of the seal effect requires a proof that $\Impl{isort}$ is abstractly upper-bounded by the $n^2$-cost specification~\citep[Example 3.4]{grodin-niu-sterling-harper>2024}.
\end{example}

Analytically, a proof that an operation $f$ on a data structure meets the desired upper bound is an inequality as below---that is, a \emph{lax commutative square}---showing that the concrete implementation $\TOP{f}$ is upper-bounded by the specification $\ABS{f}$ relative to the abstraction functions $\alpha$ and $\beta$:
\begin{center}
  \begin{minipage}{0.6\linewidth}
    \[ \left( \bindex{\TOP{f}~\TOP{x}}{\TOP{y}} \ret{\beta~\TOP{y}} \right) \le \ABS{f}~(\alpha~\TOP{x}) \]
  \end{minipage}%
  \begin{minipage}{0.4\linewidth}
    \[\begin{tikzcd}
      {\TOP{X}} & {\U{\F{Y_{\top}}}} \\
      {\ABS{X}} & {\U{\F{Y_{\abs}}}}
      \arrow["{\TOP{f}}", from=1-1, to=1-2]
      \arrow[""{name=0, anchor=center, inner sep=0}, "\alpha"', from=1-1, to=2-1]
      \arrow[""{name=1, anchor=center, inner sep=0}, "{\U{\F{\beta}}}", from=1-2, to=2-2]
      \arrow["{\ABS{f}}"', from=2-1, to=2-2]
      \arrow["\ge"{description}, draw=none, from=1, to=0]
    \end{tikzcd}\]
  \end{minipage}%
\end{center}
Letting $X \isdef \Glue{\Op_{\abs}{\ABS{X}}}{\Cl_{\abs}{\TOP{X}}}{\Cl_{\abs}(\etaOp \circ \alpha)}$ and $Y \isdef \Glue{\Op_{\abs}{Y_{\abs}}}{\Cl_{\abs}{Y_{\top}}}{\Cl_{\abs}(\etaOp \circ \beta)}$,
we will assemble the above data into a phased function $f : \U{X \pto \F{Y}}$ using the sealing effect.
\begin{lemma}\label{lem:lex}
  The following equivalence holds, since $\Op_{\abs}$ and $\Cl_{\abs}$ are lex \citep[\S 3]{rijke-shulman-spitters>2020}:
  \[ X = \sum_{\BEH{x} : \Op_{\abs} \ABS{X}} \Cl_{\abs} \left(\sum_{\TOP{x} : \TOP{X}} \Op_{\abs} (\alpha~\TOP{x} = \BEH{x}~\b)\right). \]
\end{lemma}
Let $\Impl{spec} \isdef \bindex{\ABS{f}~(\BEH{x}~\b)}{\ABS{y}} \ret{\starCl~\b, \etaOp\ABS{y}}$ run the specification $\ABS{f}$ in the abstract phase (\ie, in context $\b : \abs$).
Up to \cref{lem:lex}, the definition of $f$ is as follows:
\iblock{
  \mrow{f : \U{X \pto \F{Y}}}
  \mhang{f~(\BEH{x}, \etaCl (\TOP{x}, h)) \isdef \seal{\Impl{impl}}{\Impl{spec}}}{
    \mrow{\kw{where}~\Impl{impl} \isdef \bindex{\TOP{f}~\TOP{x}}{\TOP{y}} \ret{\etaCl\TOP{y}, \etaOp(\beta~\TOP{y})}}
  }
  \mrow{f~(\BEH{x}, \starCl~\b) \isdef \Impl{spec}}
}\noindent
This code combines the information of the functions $\TOP{f}$ and $\ABS{f}$ and the above lax commutative square into a single function $f$ that concretely runs $\TOP{f}$ and abstractly runs $\ABS{f}$.
\begin{enumerate}
  \item
    In the $\etaCl (\TOP{x}, h)$ case the implementation $\TOP{f}$ is sealed against the specification $\ABS{f}$, up to modalities.
    The proof $h : \Op_{\abs} (\alpha~\TOP{x} = \BEH{x}~\b)$ is used in the proof that abstractly, $\Impl{impl} \le \Impl{spec}$:
    \begin{align*}
      &\bindex{\TOP{f}~\TOP{x}}{\TOP{y}} \ret{\etaCl\TOP{y}, \etaOp(\beta~\TOP{y})} \\
      &= \bindex{\TOP{f}~\TOP{x}}{\TOP{y}} \ret{\starCl~\b, \etaOp(\beta~\TOP{y})} \tag*{(abstract phase)} \\
      &= \bindex{(\bindex{\TOP{f}~\TOP{x}}{\TOP{y}} \ret{\beta~\TOP{y}})}{\ABS{y}} \tag*{(monad assoc.)} \ret{\starCl~\b, \etaOp\ABS{y}} \\
      &\le \bindex{\ABS{f}~(\alpha~\TOP{x})}{\ABS{y}} \ret{\starCl~\b, \etaOp\ABS{y}} \tag*{(lax commutative square)} \\
      &= \bindex{\ABS{f}~(\BEH{x}~\b)}{\ABS{y}} \ret{\starCl~\b, \etaOp\ABS{y}} \tag*{($h$)}
    \end{align*}
  \item
    The $\starCl~\b$ case occurs in the phase, so it produces the same $\Impl{spec}$.
    Since $\seal{a}{\BEH{a}} = \BEH{a}$ in the abstract phase, these two cases cohere, as is required by the concrete modality.
\end{enumerate}

\begin{example}
  Let us instantiate this pattern to the batched queue construction of \cref{sec:abstraction:gluing:queue}, where $\TOP{X} = \listty{E} \times \listty{E}$, $\ABS{X} = \listty{E}$, and $\alpha = \Impl{revAppend}$ as before.
  The batched dequeue implementation $\TOP{\Impl{dequeue}}$ takes either constant time (in the common case) or linear time (when a ``batch'' of elements is moved from the inbox to the outbox).
  Thus, the abstract program
  \begin{align*}
    &\ABS{\Impl{dequeue}} : \U{\listty{E} \pto \F{E \times \listty{E}}} \\
    &\ABS{\Impl{dequeue}}~l \isdef \charge{\len{l}}{\ret{\Impl{uncons}~l}}
  \end{align*}
  can be used as an \emph{upper bound} specification on the batched dequeue implementation, even though it will rarely coincide with the true cost.
  Letting $X$ be the glued type of batched queues, the operation
  \[ \Impl{dequeue} : \U{X \pto \F{E \times X}} \]
  may be defined by instantiating the above template with the following lax commutative square:
\begin{center}
  \begin{minipage}{0.67\linewidth}
    \[ \left( \bindex{\TOP{\Impl{dequeue}}~\TOP{x}}{e,\TOP{x}'} \ret{e, \alpha~\TOP{x}'} \right) \le \ABS{\Impl{dequeue}}~(\alpha~\TOP{x}) \]
  \end{minipage}%
  \begin{minipage}{0.33\linewidth}
    \[\begin{tikzcd}
      {\TOP{X}} & {\U{\F{E \times \TOP{X}}}} \\
      {\ABS{X}} & {\U{\F{E \times \ABS{X}}}}
      \arrow["{{\Impl{dequeue}}_{\top}}", from=1-1, to=1-2]
      \arrow[""{name=0, anchor=center, inner sep=0}, "\alpha"', from=1-1, to=2-1]
      \arrow[""{name=1, anchor=center, inner sep=0}, "{\U{\F{E \times \alpha}}}", from=1-2, to=2-2]
      \arrow["{\Impl{dequeue}_{\abs}}"', from=2-1, to=2-2]
      \arrow["\ge"{description}, draw=none, from=1, to=0]
    \end{tikzcd}\]
  \end{minipage}%
\end{center}
  This is analogous to the squares of \cref{sec:abstraction:gluing:queue}, but now accounting for both cost and behavior.
\end{example}

The sealing effect also interoperates smoothly with the abstract phased quotients of \cref{sec:abstraction:quotient}: branches of code may use sealing against an upper bound in order to respect the abstract quotient.
\section{Conclusion}\label{sec:conclusion}

We offer a unifying perspective on modular verification of cost and behavior in dependent type theory.
The synthetic phase distinction between concrete and abstract aspects, later extended to algorithmic and behavioral aspects, internalizes the separation between implementation and specification.
At the heart of this work lies the \emph{fracture and gluing} theorem (\cref{thm:fracture}), which---underlying the slogan \emph{abstraction functions as types}---renders every type as having a concrete component, an abstract component, and a phased abstraction function relating the two.
In \cref{sec:abstraction} we present two techniques for using the abstract phase to conceal implementation details while exposing the abstract behavior to clients: the glue type $\Glue{\BEH{X}}{\ALGO{X}}{\chi}$, which encapsulates an abstraction function within the type itself, and phased quotient types, which identify concrete distinctions between representations in the abstract phase.

Critically, the \emph{noninterference} property of the phase distinction (\cref{thm:noninterference}) ensures that concrete data cannot influence the abstract view, thereby enabling modularity~(\cref{cor:modularity}).
This form of modularity allows for the compositional verification of algorithms and data structures via phased interfaces, as shown in \cref{sec:interface}.
These interfaces permit downstream code to be verified against an abstract specification using any convenient-to-reason-about library implementation.
The ability to verify against a specification abstractly while preserving freedom in concrete implementation is crucial for modular verification.

In \cref{sec:cost} we extend this approach to account for cost verification, generalizing the abstract phase into a behavioral phase and the dependent type theory to Calf, a cost-aware, effectful dependent type theory.
Under the abstract phase, programs are treated as cost-and-behavior specifications; under the behavioral phase, which redacts cost information, programs are treated as their functional behavior.
These phased interfaces enable downstream verification of both cost and functional correctness without requiring visibility into implementation details.
Noninterference guarantees extend naturally to these phases, ensuring that modular reasoning principles apply uniformly across cost and behavior.
Finally, recognizing that a cost specification may be an over-approximation of the actual implementation, we introduce a sealing effect in Decalf, a directed variant of Calf.
This effect relaxes concrete programs to cost upper bounds in the abstract phase, which allows abstract specifications to classify not only the programs that exactly match the cost, but more generally programs that are upper-bounded by the specification; this ability is crucial for modularity, since the true cost of a concrete implementation often depends on publicly-inaccessible details.
The semantics of the sealing effect is given by a monadic structure built atop a specialized glue type that accommodates inequality.

Altogether, synthetic phase distinctions provide a principled and expressive language that reconciles abstraction, efficiency, and verification in dependent type theory.

\subsection{Related Work}

We now summarize the relationship between our development and prior work.

\subsubsection{Synthetic Phase Distinctions}
This work is set in the world of synthetic phase distinctions pioneered by \citet{sterling-harper>2021,sterling-harper>2022}.
The techniques used are built upon the framework for modalities developed by \citet{rijke-shulman-spitters>2020}, making particular use of the open and closed modalities associated with the propositions $\abs$ and $\beh$ and the corresponding fracture and gluing.

Recent work by \citet{gratzer-sterling-angiuli-coquand-birkedal>2022} on abstraction in dependent type theory also makes use of phases to selectively reveal the implementation of definitions; within the phase for a particular definition, the corresponding code is unfolded. The role of phase there is to hide the implementation details of a definition; we achieve a similar effect by describing concrete data as closed-modal.

As regards the connections to cost analysis, the principal references are \citet{niu-sterling-grodin-harper>2022} and \citet{grodin-niu-sterling-harper>2024} on which \cref{sec:cost} is based.
Therein are provided a comprehensive comparison to related work on formalized cost analysis, all of which applies as well to the present setting.

\subsubsection{Realignment and the Strict Glue Type}
In this work we consider the glue type $\Glue{\BEH{X}}{\ALGO{X}}{\chi}$, which is provably equivalent to $\BEH{X}$ in the abstract phase by \cref{thm:fracture}.
This is important for an implementation to match its specification, as shown in \cref{sec:interface}.
A similar result can be achieved by considering the \emph{realignment/strictification axiom} \cite{birkedal-bizjak-clouston-grathwohl-spitters-vezzosi>2016,orton-pitts>2016,sterling>2022-logical-relations,sterling>thesis} and an aligned $\Sigma$ type called the \emph{strict glue type} \cite{sterling-harper>2022,yang>thesis,li-harper>2025} that is judgmentally equal to its abstract component under the phase.

\subsubsection{Ghost Code}
Our use of phases broadly parallels the technique of \emph{ghost code}, where functional, specification-level ghost code is maintained alongside (typically more efficient) ``regular'' code.
Prior accounts of ghost code have described noninterference of ghost code with regular code, erasing ghost code~\citep{owicki-gries>1976,filliatre-gondelman-paskevich>2016} or type refinements~\citep[\S8.2]{sterling>thesis} to extract the efficient regular code.
Although our presentation supports the extraction of concrete code as an external notion, achieved in \cref{mod:false} by interpreting $\abs$ as the false proposition, the directionality of our phase is dual: internally, we allow erasure of regular (concrete) code, leaving behind only abstract specification.
This ensures our opposite variety of noninterference, of regular code with ghost code (\cref{thm:noninterference}), which appears here as the essence of modular verification:
although ghost code and type refinements must not affect the underlying program, in modular development private/concrete data must not appear in public/abstract specifications.

\subsubsection{Representation Independence and Univalence}
\citet{angiuli-cavallo-mortberg-zeuner>2021} tell a similar story for abstract data types and representation types in a univalent setting.
For example, in their presentation of batched queues, the pair-of-lists type is quotiented by equivalence under $\Impl{revAppend}$, leading to a type equivalent to the $\listty{\nat}$~\citep[\S 4.2]{angiuli-cavallo-mortberg-zeuner>2021}.
We adapt their work to the phased setting in \cref{sec:abstraction:quotient}, activating the quotient only in the abstract phase.
That way, the true code (without a quotient) may be recovered in \cref{mod:false} by interpreting $\abs$ as the false proposition, but the verification advantages developed by \citeauthor{angiuli-cavallo-mortberg-zeuner>2021} are available internally in the abstract phase.

\subsubsection{Verification of Data Structures via Abstraction Functions}
This work is far from the first to verify data structures using abstraction functions; for example, \citet{nipkow>2025} has developed an extensive suite of data structures in Isabelle/HOL with verifications based on abstraction functions.
\citet{nipkow>2016} also shows how to verify concrete tree-based data structures using the in-order traversal abstraction function as we have in \cref{sec:abstraction:gluing}.
The abstract phase elevates abstraction functions to a distinguished status, existing within types themselves, and it provides the ability to uniformly apply all the abstraction functions for abstract, interface-level reasoning.

Moreover, using different background colors, \citet{nipkow>2025} distinguishes between ``functional programs'' (concrete) and ``auxiliary definitions'' (abstract).
Our development with the abstract phase can be viewed as a formalization of this distinction.

\subsubsection{Algebraic Specification}
In programming languages for refined modular development, such as Extended ML~\citep{kahrs-sannella-tarlecki>1997}, equational laws may be specified within an interface.
However, as discussed in \cref{sec:interface:adt:properties}, such properties will rarely hold for efficient implementations of the interface.
\citet[\S 8]{sannella-tarlecki>2012} address this issue using a distinguished class of ``observable types''---basic types at which equations are meant to apply---analogous to of the abstract types considered in this work.
For example, by stating the monoid equations under the abstract modality, we reconstruct this technique synthetically, as abstract specifications are thought of as publicly observable.
Moreover, the behavioral phase generalizes the approach to the setting of cost analysis, where many equations even at observable types only hold when cost is suppressed.

\subsubsection{Existential Types}

In this work, incorporation of abstraction functions into types is the mechanism for enforcing modularity; a classic alternative in the non-dependent setting is the use of existential types~\citep{mitchell-plotkin>1988}.
Whereas the extension of an element of an existential type is determined by the fixed set of operations provided, abstraction functions commit to an abstract representation prior to defining \emph{any} operations.
This commitment facilitates modularity: rather than hiding data in an opaque existential package, for fear that a client might abuse private implementation details, abstraction functions choose up-front the details intended to be made public, and downstream programs are free to write functions on the private type, so long as they cohere with the choice of abstract model.
This aligns with the proposal by \citet{reynolds>1985} that abstract types should not be confined to a ``closed scope,'' as they are in the case of existential types.

\subsubsection{Views}
Our development of phased quotients in \cref{sec:abstraction:quotient} captures a similar perspective as \emph{views} \cite{wadler>1987}, where data structures can be seen as a representation of a simpler type (the \emph{view}) with improved efficiency.
Functions can be given to convert back and forth between the data structure and the view, but they will not na\"ively form an equivalence:
\begin{myquote}{wadler>1987}
  The correctness of the view depends on the equivalence between the [concrete representations]; otherwise, the\dots functions would not be inverses.
\end{myquote}
The abstract phase provides a synthetic means for distinguishing between general equivalence, for implementation-level verification, and abstract equivalence, for client-facing correctness.

\subsection{Future Work}

The abstract phase forces the programmer to carefully redact implementations to abstractly reveal only the intended information to ensure modularity.
This process distills the essence of algorithms and data structures, isolating the choices being made for efficiency.
As a next step, we intend to verify additional algorithms and data structures in this style to validate this approach.
Furthermore, this work considers only worst-case cost interfaces; in ongoing work, we are considering the incorporation of more advanced tools of complexity analysis, including amortized \citep{grodin-harper>2024}, asymptotic, and expected cost specifications.
 
\section*{Data Availability Statement}
The queue examples presented in \cref{sec:abstraction,sec:interface} and a corresponding library of phased definitions and lemmas are mechanized in the Cubical Agda proof assistant~\citep{norell>2009,vezzosi-mortberg-abel>2019}, available as an artifact~\cite{grodin-li-harper>2025-afat-agda}.

\begin{acks}
  The authors thank Yue Niu and Jonathan Sterling for fruitful adjacent collaboration that broadly inspired this research;
  the anonymous reviewers for their thoughtful comments;
  and Na\"{i}m Camille Favier, Am\'{e}lia Liao, and Tesla Zhang for their advice pertaining to the Cubical Agda formalization.

  This material is based upon work supported by the \grantsponsor{AFOSR}{United States Air Force Office of Scientific Research}{https://www.afrl.af.mil/AFOSR/} under grant numbers \grantnum{AFOSR}{FA9550-21-0009} and \grantnum{AFOSR}{FA9550-23-1-0434} (Tristan Nguyen, program manager). Any opinions, findings and conclusions or recommendations expressed in this material are those of the authors and do not necessarily reflect the views of the AFOSR.
\end{acks}

\bibliographystyle{ACM-Reference-Format}
\bibliography{main,manual}

\end{document}